\documentclass[11pt]{article}
\usepackage{amsfonts, amssymb, amsmath, mathrsfs, amsthm, float, enumitem, verbatim, bm, appendix, relsize, 
                   tabularx}
\usepackage[margin=0.75in]{geometry}
\usepackage[hang,flushmargin]{footmisc}
\numberwithin{equation}{section}
\DeclareMathOperator*{\argmax}{arg\,max}
\newtheorem*{definition}{Definition}
\newtheorem*{lemma}{Lemma}
\newtheorem{theorem}{Theorem}[subsection]

\title{Pragmatic Information, Computation,  and the Efficient Market Hypothesis}
\author{Edward D. Weinberger\\Department of Finance and Risk Engineering\\NYU Tandon School of Engineering}
\begin{document}
\maketitle
\abstract{The efficient market hypothesis, that prices reflect all available information, imputes a meaning to market moving information, a view of information foreign to standard information theory.  Here, after reviewing properties that make a proposed formula for ``pragmatic information'' a plausible measure of meaning, we consider the role that the receiver’s position in the machine hierarchy corresponding to the Chomsky hierarchy play in extracting this meaning.  We show that a receiver at a given level in the hierarchy may be unable to extract pragmatic information from a message because it appears random, yet a receiver at a higher level in the hierarchy finds the message perfectly intelligible.   Also, the maximum processing rates for messages of different levels of hierarchy serve as a kind of channel capacity, leading to a tradeoff between the amount of pragmatic information extracted and the extraction time.  All of the above suggests a recasting of market efficiency in terms of ``computational efficiency'', i.e. the question of whether a market appears efficient to a participant with a given computational endowment.   Successful trading strategies implementable as finite state machines, the lowest level in the hierarchy, imply departures from ``finite state efficiency''.  We show, via a stylized example, how pragmatic information can characterize the ensuing approach to finite state efficiency.  We also show that actual market dynamics can be more computationally intractable than any finite state machine can process by proving the PSPACE-completeness of the processing of the smart order routing systems of major brokerage firms.  We conclude that computational efficiency is the norm and pragmatic information can characterize any departure from it.
} 

\pagebreak

\parindent 0px
\maketitle
\section{Introduction}
The standard measure of the amount of information in a transmitted message says nothing about the ``meaningfulness'' of the message.  However, numerous authors have attempted to put forth a theory of ``pragmatic information'' that addresses precisely that missing measure of meaning.  One such paper, \cite{Weinberger24}, proposes that a message acquires meaning from a change in the receiver's prior beliefs, as expressed by the likelihood that the receiver makes a particular decision.  If so, the computational abilities of the receiver have much to say about the the amount of pragmatic information that it can process,
the primary topic of the present paper.  To make this presentation, we first review the definition of pragmatic information in \cite{Weinberger24}, a few of its relevant properties, and some relevant terminology.  Even before that, however, we contrast the theory presented in \cite{Weinberger24} with the subject commonly known as ``information theory''.  As Warren Weaver famously observes \cite{ShannonAndWeaver62}, the effectiveness of a communications process could be measured by answering any of the following three questions: \\
\begin{enumerate}
\item[A.] How accurately can the symbols that encode the message be transmitted (``the technical problem")?
\item[B.] How precisely do the transmitted symbols convey the desired meaning (``the semantics problem")?
\item[C.] How effective is the received message in changing conduct (``the effectiveness problem")?\\
\end{enumerate}
It is, by now, well known that ``standard'' information theory concerns itself only with the answer to above question A.  In other words, the standard theory distinguishes only between correctly and incorrectly transmitted symbols, regardless of their meaning;
thus, the standard theory makes no distinction between the correct transmission of a hot tip about the next big AI stock and the correct transmission of
a jumble of random letters.  Pragmatic information, in contrast, distinguishes between the stock tip and the random jumble by virtue of the fact that the former will likely lead the 
receiver to buy the stock (the ``changed conduct'' of question C, above) and the latter will not\footnote{A particularly compelling example of the distinction between mere transmission of symbols and the conveyance of 
meaning is found in an apocryphal story about the
great French writer Victor Hugo, who left for vacation immediately after submitting the manuscript of his masterpiece {\it Les Mis\'erables} to his publisher and before the actual publication date of the novel.  Not wanting to write much during his vacation, but thirsting for news about the reception of his work, Hugo sent his publisher a letter consisting of the single character ``?'', to which the publisher indicated the book's commercial success by replying simply ``!''.}.\\

Thermodynamics makes the distinction between the so-called ``free energy'' of a system, or the 
energy available to do work and the total energy of a system, which also includes, for example, heat energy.  
\cite{Weinberger24} makes the case that the same distinction should be made between pragmatic information and the total 
information of a message.  \cite{CrutchfieldID} made a similar distinction by introducing the term ``free information'' to 
describe essentially the same quantity, in the special case described in the ``Partitions of $\mathcal M$'' section below.\\

One motivation for this work is that it might help understand the efficient market hypothesis of financial economics, as will be seen in a subsequent section of this paper.  
Other applications will, no doubt, present themselves to the mind of the reader 
(One such is the intriguing finding that there is an average rate at which information is conveyed by human language,  independent of the speaker or the language spoken \cite{SpeakingRate}, suggesting an intrinsic neurological processing speed
of the kind we discuss below.).\\

The next section reviews our definition of pragmatic information, some of its more important properties, and some useful terminology, 
all of which are set forth in more detail in \cite{Weinberger24}.   Our primary purpose in presenting these results is to preface this paper's subsequent sections:  it is only by convincing readers that our notion of pragmatic information really is a ``measure of meaning'' that they might care about what
these subsequent sections have to say about the role that the receiver's computational abilities play in extracting this meaning
and the relevance of all of this to the efficient market hypothesis.

\section{A Definition of Pragmatic Information and Some Results that Justify It}

If the practical meaning of a message stems from its usefulness in making decisions, the message must do so by changing those beliefs about the state of the world that will inform the subsequent action.  
Accordingly, we consider a decision maker, $\Delta$, which may be a machine, a human, or a non-human organism.  We assume that $\Delta$'s decisions are informed by beliefs about a random variable, $\omega$, 
which we assume to have the discrete sample space $\Omega =\{\omega_1, \omega_2, \ldots, \omega_i,\ldots, \omega_N \}$.  
$\Delta$ has {\it a priori} beliefs about $\omega$, reflected by $\Delta$'s assignment of the
 positive, prior probabilities  $\mathbf q = ( q_1, q_2, \ldots, q_N)$ to the various possible values of $\omega$.  We
refer to $\mathbf q$ as prior probabilities because we then supply $\Delta$ with a ``message", $m$, 
of arbitrary length, chosen from some ensemble of possible messages ${\mathcal M}$, each with probability $\varphi_m$.  
As a result, $\Delta$ updates its beliefs about the state of the world, reflected in a change from $\mathbf q$ to $\mathbf p_m = 
( p_{1|m}, p_{2|m}, \ldots, p_{N|m})$.  We assume that the update is made via an abstract automaton, $\mathfrak A_\Delta$, 
though we leave open the 
possibility that the update is actually implemented by a human.  The update might involve the interaction of $m$ with internal 
data, $\mathfrak d$, which is already stored within $\Delta$ and which is chosen 
from some ensemble of possible data items, ${\mathfrak D}$. We consider only the simplest case, in which $\mathfrak d$ is chosen once 
and for all and is always the same for all input messages $m$.  We will write $m^+$ when we want 
to explicitly include $\mathfrak d$, as we will when we want to consider the length of the data stream 
presented to $\Delta$.\\  

Given the above, \cite{Weinberger24}  proposed the following\footnote{While we don't think that the 
exact definition proposed in \cite{Weinberger24} has appeared in the literature before, close 
cousins of it have.  See the discussion in \cite{Weinberger24} where 
it is emphasized that the new idea in \cite{Weinberger24} is that $\Delta$'s prior beliefs about $\omega$, embodied in $\mathbf q$ 
can be just plain wrong.}

\begin{definition}
The {\bf pragmatic information}, ${\Phi_\Delta (\mathcal M; \Omega)}$  of an ensemble\footnote{In Section 4 of this paper, we allow for the possibility that this ensemble can be infinite.} of messages ${\mathcal M}$ is 
\begin{align}
 \Phi_\Delta (\mathcal M; \Omega) &= \sum_{i, m} \varphi_m  p_{i|m} \log \left(\frac{p_{i|m}}{q_i}\right)  \label{eq: BasicDef}\\
                                                  &= \sum_{i, m}  p_{i,m} \log \left(\frac{p_{i|m}}{q_i}\right),
\end{align}
where all logarithms are assumed to be base 2 logarithms and we take $p_{i|m} \log \left(\frac{p_{i|m}}{q_i}\right) = 0$ if $p_{i|m} =0$.
\end{definition}

\begin{definition}
The {\bf conditional pragmatic information} of messages in ${\mathcal M'}$ acting upon $\Delta'$, given that $m \in \mathcal M$ is acting on $\Delta$, as well as the resulting outcome $\omega$, is
\begin{align*}
 \Phi_{\Delta'| \Delta} (\mathcal M'; \Omega' | \mathcal M, \Omega) = \sum_{i, i', m, m'} p_{i, i', m, m'} \log \left(\frac{p'_{i'|i,m,m'}}{q_{i'|i}}\right) 
\end{align*}
\end{definition}

\cite{Weinberger24} noted several features of these definitions that make pragmatic information both an information measure and a plausible measure of the amount of meaning in the messages in ${\mathcal M}$:
\begin{itemize}
\item $\Phi_\Delta (\mathcal M; \Omega) \geq 0$. with $\Phi_\Delta (\mathcal M; \Omega) > 0$ unless $\mathbf p_m = \mathbf q$ for all $m \in \mathcal M$.
\item When $\mathbf p_m = \mathbf u_k$, {\it i.e.} a unit vector in the 
$k^{\rm th}$ direction for some $k$, we say that $m$ is {\it pragmatically definitive}.   $\Phi_\Delta (\mathcal M; \Omega)$ 
achieves its maximum among the pragmatically definitive $\mathbf p_m$'s, 
assuming that there are any.  If so, for a given set of prior probabilities, $\mathbf q$, let 
$k(m) = \argmax (-\log q_l)$, with the max taken over all of the $q_l$'s for which $p_{l|m} > 0$.  The resulting maximum is then 
\begin{align}
\Phi_\Delta (\mathcal M; \Omega) \leq -\sum_m \varphi_m \log q_{k(m)}.   \label{eq: MaxPragInfo}
\end{align}
\item If the decisions of makers $\Delta$ and $\Delta'$ are, respectively, informed by messages $m \in \mathcal M$ and $m' \in 
\mathcal M'$,  with respective prior probabilities $\mathbf q$ and $\mathbf q'$ and with corresponding {\rm a posteriori} 
probabilities $\mathbf p_m$ and $\mathbf p'_{m'}$, then
$$
 \Phi_{\Delta, \Delta'} (\mathcal M; \Omega, \Omega')=\Phi_{\Delta} ({\mathcal M}; \Omega)  +
\Phi_{\Delta'| \Delta} (\mathcal M'; \Omega' | \mathcal M, \Omega).
$$
If, in addition, 
\begin{align}
q'_{i'|i} = q'_{i'} {\rm \quad and \quad} p_{i, i', m, m'} = p_{i, m} p'_{i', m'}  {\rm \ for \ all \ } i, i', m, m', 
\end{align}
then 
\begin{align}
 \Phi_{\Delta, \Delta'} (\mathcal M, \mathcal M'; \Omega, \Omega')=\Phi_{\Delta} ({\mathcal M}; \Omega)  +
\Phi_{\Delta'} ({\mathcal M'}; \Omega')   
\end{align}
In other words, pragmatic information is additive for probabilistically independent message ensembes and completely 
independent decision makers.  However, mere probabilistic independence of the message ensembles is insufficient to guarantee 
pragmatic independence.
\item $\Phi_\Delta (\mathcal M; \Omega)$ is the expected number of extra bits required for the binary encoding of samples from $\mathbf p_m$, averaged over $m \in \mathcal M$, using an encoding optimized for $\mathbf q$, rather than,
for each $m$, an encoding optimized for that $\mathbf p_m$.  In other words, $\Phi_\Delta (\mathcal M; \Omega)$ is the expected number of bits that $\Delta$ has learned by processing $m \in \mathcal M$.
\item Per Theorem 5.0.1 of \cite{Weinberger24}, if $\bm \phi$ is the vector of probabilities  whose $i^{\rm th}$ component is 
given by $\phi_i = \sum_{m \in \mathcal M} p_{i,m}$ for all $i$, $D_\Delta(\bm \phi || \mathbf q)$ is the (non-negative) Kullback-Leibler divergence of $\bm \phi$ and $\mathbf q$, and
$\mathcal I (\mathcal M; \Omega)$ is the mutual information between $\mathcal M$ and $\Omega$,
\begin{align}
\Phi_\Delta({\mathcal M}; \Omega) &= \mathcal I (\mathcal M; \Omega) + D_\Delta(\bm \phi || \mathbf q) \cr
                                                 &= \mathcal H (\mathcal M) - H_{\Delta} (\mathcal M| \Omega) + 
                                                       D_\Delta(\bm \phi || \mathbf q) \label{eq: PhiDecomp} \cr
									&\geq \mathcal I (\mathcal M; \Omega),
\end{align}
with equality if and only if $\bm \phi = \mathbf q$.

\end{itemize}

\cite{Weinberger24} makes the case that it is often reasonable to assign a numerical {\it value function},
$V: \mathcal M \to \Re$ such that $\mathcal M = \mathcal D_\Delta \cup \mathcal I_\Delta \cup 
\mathcal U_\Delta$,\footnote{See \cite{Weinberger24} for a more complete discussion of this decomposition} where: 
\begin{itemize}
\item $\mathcal D_\Delta$, the set of pragmatically disinformative messages with respect to $\Delta$, is defined by 
$$
\mathcal D_\Delta = \{m \in \mathcal M: V(m) < 0\},
$$
\item $\mathcal I_\Delta$, the set of pragmatically irrelevant messages with respect to $\Delta$, is defined by 
$$
\mathcal I_\Delta = \{m \in \mathcal M: V(m) = 0\},
$$
\item $\mathcal U_\Delta$the,  set of pragmatically useful messages with respect to $\Delta$, is defined by 
$$
\mathcal U_\Delta = \{m \in \mathcal M: V(m) > 0\}.
$$
\end{itemize}
Intuitively, we think of the processing of messages in $\mathcal D_\Delta$, $\mathcal I_\Delta$, and $ \mathcal U_\Delta$ as leaving $\Delta$ respectively worse off, the same, or better off than before.
\cite{Weinberger24} defines $\mathcal N_\Delta = \mathcal I_\Delta \cup \mathcal U_\Delta$ as the set 
of messages that are ``pragmatic noise''.  Evidently, 
$$
\Phi_{\Delta} ({\mathcal M}; \Omega) = \Phi_{\Delta} (\mathcal 	D_\Delta; \Omega)  
												+ \Phi_{\Delta} (\mathcal I_\Delta; \Omega) 
                                                                 + \Phi_{\Delta} (\mathcal U_\Delta; \Omega)
$$
Thus, if $\Phi_{\Delta} (\mathcal M; \Omega)$ is zero or negligibly small, so must be $\Phi_{\Delta} (\mathcal U_\Delta; 
\Omega)$, a fact that will play an important role in the sequel.

\section{Partitions of $\mathcal M$}
We anticipate that $|\mathcal M| >> |\Omega|$ in many applications of \eqref{eq: BasicDef}.  In such cases, $\Delta$ must respond identically to 
many of the messages in $\mathcal M$, thus partitioning $\mathcal M$ into a set, $\Pi$, of equivalence classes.  We then have the following 
\begin{theorem} 
Even though $\Phi_{\Delta}(\mathcal M; \Omega)$ is bounded by \eqref{eq: MaxPragInfo}, the Shannon entropy, $\mathcal H(\mathcal M)$, can be  arbitrarily large.
\end{theorem}
\begin{proof}
Suppose, 
for example, $\mathcal A$ is the alphabet consisting of the symbols $\mathtt P$, $\mathtt Q$, $\mathtt R$, and the special ``end of string'' symbol, $\mathtt{END}$, and suppose 
${\mathcal M} = \mathcal A^*$, the set of  all strings consisting of these symbols.  Suppose further that $\Delta$ 
outputs $\omega_0$ if $m$ consists entirely of $\mathtt P$'s or $\mathtt Q$'s, except for a terminal $\mathtt{END}$ but are 
otherwise arbitrary elements of $\mathcal M$, and outputs $\omega_1$ otherwise.  
For arbitrarily small $\epsilon > 0$ and an arbitrarily large integer $L$, assign probabilities to all messages $m \in \mathcal M$ as
$$
\varphi_m = \begin{cases} (1-\epsilon)/T, & \quad \text{if $|m| < L$} \\
                                        2 \times 3^{L - 1 - 2|m|}\epsilon,& \quad \text{otherwise}
                   \end{cases}
$$
where 
$$
T = \sum_{k=0}^{L-1} 3^k,
$$
is the number of strings in $\mathcal M^*$ of length less than $L$.  We then compute
\begin{align*}
\mathcal H(\mathcal M) &= - \sum_{m \in \mathcal M} \varphi_m \log \varphi_m   \\
                                   &= - \sum_{|m| < L} \varphi_m \log \varphi_m - \sum_{|m| \geq L} \varphi_m \log \varphi_m   \\
                                   &> (1-\epsilon) \log\left[\frac{T}{1-\epsilon}\right].
\end{align*}
There are $T$ of the $\varphi_m$'s for $|m| < L$, all of which have the identical value 
$(1-\epsilon)/T$, so the expression in the third line above is the value of the first sum in the line above it.  The third line follows from the second because ignoring the second sum in the second line decreases the value of the result.  
Since $L$ and thus $T$ can be made arbitrarily large, so can $\mathcal H(\mathcal M)$.\\
\end{proof}

Evidently, $\mathbf p_m= \mathbf p_{m^*}$ for all  $m, m^*$ in the same equivalence class $\pi \in \Pi$.  If we denote this vector of probabilities that is common to all $m \in \pi$ as $\mathbf p_\pi = (p_{1|\pi}, p_{2|\pi}, \ldots p_{N|\pi})$ 
and the sum $\sum_{m \in \pi} \varphi_m$ as $\varphi_\pi$, then 
\begin{align*}
\Phi_\Delta (\mathcal M; \Omega) &= \sum_{i, m} \varphi_m  p_{i|m} \log \left(\frac{p_{i|m}}{q_i}\right)  \\
                                               &= \sum_i \sum_{\pi \in \Pi} \sum_{m \in \pi} \varphi_m  p_{i|m} \log \left(\frac{p_{i|m}}{q_i}\right)  \\
                                               &= \sum_i \sum_{\pi \in \Pi} \varphi_\pi  p_{i|\pi} \log \left(\frac{p_{i|\pi}}{q_i}\right). 
\end{align*}

A different decision maker, $\Delta'$, could partition $\mathcal M$ differently, perhaps increasing  the pragmatic information that this decision maker 
could extract.  However, \eqref{eq: MaxPragInfo} suggests that the maximum pragmatic information 
extracted by any such $\Delta'$ has the following upper bound:
\begin{align}
\Phi_{\Delta^*} (\mathcal M; \Omega) \leq -\sum_\pi \varphi_\pi \log q_{k(\pi)}   \label{eq: MaxPartitionPragInfoInfo},
\end{align}
where $\Delta^*$ assigns the messages in each partition $\pi$ deterministically to 
the outcome $k(\pi) = \argmax (-\log q_l)$, with the max taken over all of the $q_l$'s for which $p_{l|m} > 0$.  \\

A finer partition $\mathcal M$ creates a larger set of equivalence classes, $\Pi'$, resulting in estimates of a random variable 
$\omega'$ that would take values in a superset, $\Omega'$ of $\Omega$, with $\Delta'$ replacing $\Delta$ as the decision
maker.  In each partition $\pi' \in \Pi'$, the {\it a posteriori} probabilities, $\mathbf q$ and $\mathbf p_m$, would then be 
replaced, for each $m \in \mathcal M$, 
by the possibly larger sets of probabilities $\mathbf q'$ and $\{p'_{j|m}\}'$, respectively.  
We then have the following
\begin{theorem} 
For all such refinements,
$$
\Phi_{\Delta'}(\mathcal M; \Omega') \geq \Phi_{\Delta}(\mathcal M; \Omega),  
$$
with equality only if the ratios $p_{j|m}/q_j$ are constant across $j \in J(i)$ for all $i$.
\end{theorem}
\begin{proof}
If $\Delta'$ induces a finer partition than $\Delta$, then, for each $m$ and for each $i$,  there must be some subset, $J(i)$, of the $j$'s such that
$$
p_{i|m} = \sum_{j \in J(i)} p'_{j|m} \mathrm{\ \ and\ \ } q_i = \sum_{j \in J(i)} q'_j
$$
We must therefore have, for all $i$ and $m$, that
$$
 \sum_{j \in J(i)} p'_{j|m} \log \left( \frac{p'_{j|m}}{q_j} \right) \geq
p_{i|m} \log \left( \frac{p_{i|m}}{q_i} \right),
$$
with equality only if the ratios $\frac{p'_{j|m}}{q_j}$ are constant across the $j$'s 
by the log-sum inequality.  Thus, for any choice of the marginal probabilities $\varphi_m$,
\begin{align*}
\Phi_{\Delta'} ({\mathcal M}; \Omega') &= \sum_{i,m} \sum_{j \in J(i)} \varphi_m p_{j|m} \log \left( \frac{p_{j|m}}{q_j} \right) \cr
                                                        &\geq \sum_{i,m} \varphi_m p_{i|m} \log \left( \frac{p_{i|m}}{q_i} \right) \cr
                                                        &=  \Phi_{\Delta} ({\mathcal M}; \Omega),
\end{align*}
with equality only if the ratios $\frac{p'_{j|m}}{q_j}$ are constant across the $j$'s for each $J(i)$.
\end{proof}
Nevertheless, merely increasing the size of $\Omega$ does not necessarily increase 
$\Phi_{\Delta} ({\mathcal M}; \Omega)$ if the hypotheses of the theorem are violated, as 
is shown by the 
following example:  Suppose we have a single $m$ and 
two outputs $\omega_1$ and $\omega_2$, with $q_1 = q_2 = 1/2$, $p_{1| m} = 1/4$, and $p_{2| m} = 3/4$, so that $\Phi_{\Delta'} ({\mathcal M}; \Omega) \approx 0.1887$.  We then introduce a third output, $\omega_3$ with $p_{3| m} =  q_3 = 0.01$, and make the adjustments $q_1 = q_2 = {{1 - 0.01} \over 2}$, $p_{1| m} = {{1 - 0.01} \over 4}$, and $p_{2| m} = {{3(1 - 0.01)} \over 4}$.  $\Phi_{\Delta} ({\mathcal M}; \Omega)$ then assumes a smaller value of approximately 0.1868.  
In other words, the addition of the new output $\omega_3$ here simply confuses things.\\

Such partitions can be viewed as a higher level description of the input message ensemble.  Indeed, \cite{CrutchfieldID} 
bases its discussion about ``free information'' on a similar partitioning, albeit in the special case that the 
decision  to be made is whether, in our notation, one subsequence of $m$, presented as a time series, 
is identical to another that appears later.  \cite{CrutchfieldID} 
presents its ideas in the language of symmetries, which we, too, can adopt by identifying the $m$'s in the same partition as
``symmetrical''.  However, the present work 
makes clear what \cite{CrutchfieldID} leaves out; namely, the question of how those symmetries arise in the first place.  Our 
answer is clear:  messages are symmetrical because our decision maker, $\Delta$, has explicitly said so.  
It is only in the context of 
the need to make a decision, if only in some generalized sense, that a message has meaning, and it is only in such a context 
that two messages can have the same meaning.\\
 
In  \eqref{eq: PhiDecomp}, we added the 
subscript $\Delta$ to the conditional Shannon entropy to emphasize that the conditioning depends on the 
decision maker.  We now observe that $\mathcal H (\mathcal M)$ remains the same regardless of how $\mathcal M$ is partitioned.  
Thus, if $\Delta'$ induces the finer partition, $\Pi'$, the increase in pragmatic information can be written
\begin{align*}
\Phi_{\Delta'}({\mathcal M}; \Omega) - \Phi_{\Delta}({\mathcal M}; \Omega) &= 
                                                           \Big[\mathcal H_{\Delta} (\mathcal M| \Omega) - \mathcal H_{\Delta'} (\mathcal M| \Omega) \Big]
                                                           + \Big[D_\Delta(\bm \phi' || \mathbf q) - D_\Delta(\bm \phi || \mathbf q)\Big]                        \\
                                                           &\geq 0.
\end{align*}
The quantity in the first set of square brackets represents the improvement in mapping the different messages $m \in \mathcal M$ to the appropriate
output $\omega_i \in \Omega$; the quantity in the second set of brackets represents the improvement in correcting the initial bias in $\mathbf q$.\\

The extra effort required for the finer partition may not be worth it.  This will be the case, for example, if 
the receipt of particular $\omega$'s can almost always be traced back to messages $m$ in some strict subset of a partition, in which 
case little is learned by making the partition finer.  A version of rate distortion theory\footnote{See, for example, \cite{CoverAndThomas}
for a discussion of this important branch of standard information theory} might therefore be based on pragmatic information.\\

\section{The Effect of the Decision Maker's Computational Capacity on Pragmatic Information}
As \cite{wow} noted, if $\Delta$ has {\it no} computational abilities at all, in the sense that any message
received will be ignored, $\mathbf p_m = \mathbf q$ for all messages $m$, and $\Phi_{\Delta} ({\mathcal M}; \Omega) =0$, regardless of the length or complexity of $m$!  
This observation raises the question of just how the computational power of the receiver influences its ability to extract pragmatic information from its input messages.  A rigorous discussion of this question 
inevitably draws upon much of the theory of formal languages and automata, as we will see when we 
demonstrate the following:
\begin{itemize}
\item If $\Delta$ can only process messages at the lowest level of the Chomsky hierarchy, there is a finite upper bound on $\Phi_{\Delta}({\mathcal M}; \Omega)$ that depends only on the 
{\it a priori} assumptions of $\Delta$, regardless of the length or complexity of $m$.
\item There can be a qualitative increase in $\Phi_{\Delta}({\mathcal M}; \Omega)$ as $\Delta$ becomes able to process 
the successively more complex messages comprising the successively higher levels of the Chomsky hierarchy.
\item The more complex messages of higher levels of the Chomsky hierarchy can appear essentially random to a processor corresponding 
to a lower level of the hierarchy, even though such messages are not, in fact, random  In such cases, no 
pragmatic information can be extracted by such a processor, even though p.ragmatic information can be 
extracted when a processor of sufficient power is supplied.
\item More computational power can increase {\it pragmatically useful} information extraction without increasing pragmatic information extraction.
\item The computation of $\Phi_{\Delta}({\mathcal M}; \Omega)$ can be undecidable if the messages in $\mathcal M$ are
in the highest level of the Chomsky hierarchy.
\end{itemize}

\subsection{Some Preliminaries, Notational and Otherwise}
We begin by assuming that each input message, $m$,  is a string of symbols, $\alpha_1, \alpha_2, \ldots $, of arbitrary length,
with each $\alpha_k$ chosen from some finite ``alphabet'' $\mathcal A$.  We
denote the number of distinct symbols in $\mathcal A$ as $|\mathcal A|$ and
the set of strings of zero or more of the $\alpha$'s as $\mathcal A^*$, which therefore includes the empty string, 
$\Lambda$.  A {\it formal language} $\mathcal L$ is a subset of $\mathcal A^*$, as described in 
Appendix A.  The Chomsky hierarchy is comprised of four classes of such languages, each qualitatively 
more complex than the classes beneath it in the hierarchy. (See Appendix A or theoretical computer 
science textbooks such as \cite{Sipser} or \cite{AHU} for an introduction to this hierarchy.)
\\

It will be convenient to use the notation for substrings that will be familiar to many readers, namely that of the Python 
programming language, summarized as follows:
\begin{itemize} 
\item $m[i]$ is the $i^{\rm th}$ symbol in the string $m$, with $m[0]$ being the first symbol,  
\item $m[{i:\ }j]$ the $i^{\rm th}$ thru $j-1^{\rm th}$ symbols in the string, $m[{\ :}i]$ the beginning of the string, up to but not 
including the $i^{\rm th}$ symbol, and $m[i{:\ }]$ the end of the string, starting with the $i^{\rm th}$ symbol.
\item the notation $m_1 + m_2$ for strings $m_1$ and $m_2$ is the concatenation $m_1 m_2$.
\end{itemize}

We will sometimes need to consider the limiting case where $m$ becomes the semi-infinite sequence, $S$.  
In contrast to the statement $m \in \mathcal A^*$,
we write $S \in \mathcal A^\infty$.  We use the same substring conventions for sequences as we use for strings, with $S[0]$ being well-defined 
because $S$ is only semi-infinite.  We then define the pragmatic information, ${\Phi_\Delta (\mathcal S; \Omega)}$, for 
$S \in \mathcal S \subset A^\infty$, as $\lim_{n \to \infty} {\Phi_\Delta (\mathcal M_n; \Omega)}$,
where $\mathcal M_n = \left\{S[\ :n] : S \in \mathcal S \right\}$, when the limit exists.
\\

Algorithms that can determine whether any given string in $\mathcal A^*$ is a member of a given language $\mathcal L$
are often called abstract automata.  Such automata are discussed at length in Appendix B, below and in numerous textbooks, such as 
\cite{Sipser} and \cite{AHU}.  When an automaton can always correctly determine that string $m \in \mathcal L$,
the automaton is said to {\it recognize} $\mathcal L$ or alternatively, that the automaton is said to 
{\it accept} $m \in \mathcal L$.  A fundamental result of theoretical computer science is the collection of theorems, shown in Table 1 below and covered
by \cite{Sipser} and \cite{AHU}, which specify precisely which abstract automata can recognize which languages.\\

\begin{table}[h]
\begin{center}
\caption{The Chomsky Hierarchy and Corresponding Automata}
\begin{tabular}{|l|p{56 mm}|c|}
\hline\hline
{\bf Language} & {\bf Recognizing Automaton} & {\bf Our Automaton Symbol} \\
\hline\hline
Regular    & Finite State Machine     & $\mathfrak F$  \\
\hline
Context Free & Non-Deterministic Pushdown Automaton & $\mathfrak P$  \\
\hline
Context Sensitive & Linear Bounded Automaton & $\mathfrak B$ \\
\hline
Recursively .Enumerable & Turing Machine & $\mathfrak T$ \\
\hline
\end{tabular}
\end{center}
\end{table}

Our present interest is in a generalization of the membership problem described above; 
namely, given $m \in \mathcal M$, determine
which, if any, of the languages $\mathcal L_1, \mathcal L_2, \ldots, \mathcal L_N$ (one for each output of $\Delta$) does 
$m$ belong.  We consider only the simplest situation in which $m$ is a member of at most one of the $\mathcal L$'s.  
We also assume that all of the $\mathcal L$'s are at the same level in the Chomsky hierarchy.  Because all such languages
are closed under finite set theoretic unions, this union is also at the same level in the hierarchy as its constituents.\\

Given the message $m$, $\mathfrak A_\Delta$ first checks to see if $m \in \mathcal L_1$.  If so, $\mathfrak A_\Delta$ 
outputs the corresponding, unique
\footnote{Multiple final states of $\mathfrak A_{\Delta}$ can map to the same $\mathbf p_m$, however.}
$\mathbf p_m$.  If not,$\mathfrak A_\Delta$ checks to see if $m \in \mathcal L_2$, and, if so, 
$\mathfrak A_\Delta$ once again outputs a corresponding, unique $\mathbf p_m$.  
If $m$ is not recognized by any of the $\mathcal L$'s, 
$\Delta$ retains the prior probabilties $\mathbf q$, so that $m$ has zero pragmatic information\footnote{That unrecognized messages have zero pragmatic information corresponds to 
the ``limit of complete novelty'', postulated by \cite{Weizaeker}.}.  \\

\subsection{Limitations on $\Phi_{\Delta} (\mathcal M; \Omega)$ If $\mathfrak A_\Delta$ is a Finite State Machine}
The simplest abstract automaton in the Chomsky hierarchy (See Table 1, above.), the finite state machine, partitions its inputs into a finite number of equivalence classes.  Hence, if $\Delta$'s parsing automaton,
$\mathfrak A_\Delta$, is a finite state machine, it must induce equivalence classes on the elements of $\mathcal M$.   In this case, the 
Myhill-Nerode Theorem \cite{Sipser} guarantees that the number, $|\Pi|$, of the partitions, $\Pi$, of $\mathcal M$ described above is the number of states in $\mathfrak F$, independent of the size of $\Omega$, or, 
as we saw above, the length of the messages in $\mathcal M$.  In general, for a given $\pi$, the probabilities $p_{i|\pi}$ 
and $p_{j|\pi}$ will both be positive for $i \neq j$, reflecting the possibility that, even after $\Delta$ receives some message $m \in \pi$, there is uncertaintly as to the value of $\omega$, and thus $\Delta$'s decision.  We 
therefore have
\begin{align*}
\Phi_{\Delta} (\mathcal M; \Omega) &=    \sum_{\pi \in \Pi} \varphi_\pi \left[\sum_i   p_{i|\pi} \log \left(\frac{p_{i|\pi}}{q_i}\right)\right]. \\
\end{align*}
A straightforward application of the corollary to Theorem 2.3.4 of \cite{Weinberger24} shows that, for each $\pi \in \Pi$, the quantity in square brackets, which is $D(\{p_{i|\pi}\}, \mathbf q)$, the Kullback-Leibler 
divergence of the $p_{i|\pi}$'s and $\mathbf q$ is bounded
by $ - \log q_{k(\pi)}$, where, for each $m$, $k(\pi) = \argmax (-\log q_l)$, with the max taken over all of the $q$'s for which $p_{l|\pi} > 0$.  We have thus proven
\begin{theorem}  
\begin{align*}
\Phi_{\Delta} (\mathcal M; \Omega) &\leq  -\sum_{\pi \in \Pi} \varphi_\pi \log q_{k(\pi)} \\
                                                    &\leq   - \log q_{\rm min},
\end{align*}
where $q_{\rm min}$ is the smallest component of $\mathbf q$, regardless of the magnitude of $|m|$.
\end{theorem}

\subsection{An Intuitive Explanation of Why Increasing Computational Power Can Increase Possible Pragmatic Information 
Extraction}
Suppose the grammar describing $\mathcal M$ is one of the grammars in the Chomsky hierarchy and a particular $\Delta$
is able to glean pragmatically useful information from messages in $\mathcal M$ because the corresponding $\mathfrak A_\Delta$
has the appropriate level of computational power.  What if some other decision maker $\Delta'$ tries to process messages from
the same $\mathcal M$, but with a strictly less powerful $\mathfrak A_{\Delta'}$ (So that $\mathcal M$ is a context free language, but $\mathfrak A_{\Delta'}$ is a finite state machine, for example.)? \\

Except for regular grammars, words in every grammar in the Chomsky hierarchy require an automaton with at least one stack to parse them correctly, so the attempt to recognize $m$'s that are messages in a
context free, context sensitive, or unrestricted grammar using a finite state automaton will fail for at least one such $m$.  Similarly, the attempt to recognize $m$'s that are messages in a context sensitive or unresticted grammar using a pushdown automaton will also fail because the pushdown automaton is missing a second stack.  
Finally, if $\mathcal M$ is a subset of an unrestricted grammar that is not context sensitive,
then there must exist $m \in \mathcal M$ whose correct parsing requires a stack that grows super-linearly with $|m|$.  It follows that such parsings will also fail if $\mathfrak A_{\Delta'}$ is a bounded linear automaton.\\

What happens when a given $m$, which, for convenience we label $m^0$, cannot be parsed correctly?  
If $\Delta$ can detect the parsing failure, then $\mathfrak A_{\Delta'}$ does not change its prior estimate of the state-of-the-
world random variable $\omega$, and $\mathbf p_{m^0} = \mathbf q$.  It follows that the terms 
for which $m=m^0$ in the formula for $\Phi_{\Delta'} (\mathcal M; \Omega)$ are zero.  If, in addition, the processing of 
those $m$'s that do not require the full power of $\mathfrak A_\Delta$ remains unchanged, the corresponding terms in 
$\Phi_{\Delta'} (\mathcal M; \Omega) $ also remain unchanged.  We conclude that 
$$
\Phi_{\Delta'} (\mathcal M; \Omega) < \Phi_{\Delta} ({\mathcal M}; \Omega),
$$
where $\Delta$ is another decision maker that uses an automaton that appears in the same row of the table as $\mathcal L$.

\subsection{A More Rigorous Explanation of Why Increasing Computational Power Can Increase Possible Pragmatic Information 
Extraction}
In this section, we demonstrate that there can be qualitative differences in the abilities of automata at the different levels of the 
Chomsky hierarchy to predict the symbol $S[n]$ of a given sequence, $S$, given its predecessor symbols.  We begin by 
following \cite{CoverAndThomas}. There, the successive symbols $S[n] \in
\mathcal A$ are interpreted as the identifiers of the winners of successive horse races.  Prior to the 
$n^{\rm th}$ race, we take these winners to be random variables.  We then 
use the results of \cite{CoverAndThomas} to establish the connection to pragmatic information.
\\

To proceed, we make the following assumptions for all races:
\begin{itemize}
\item we place bets specifying that one or more of the $|\mathcal A|$ horses wins a given race, 
\item the payout of each such bet is $\mathcal O_h$ dollars per dollar bet if horse $h$ wins the race and nothing otherwise,
\item $P_h$ is the probability that horse $h$ wins the race, possibly conditioned on the outcome of previous races, 
\item we bet the fraction, $B_h$, of our wealth on horse $h$, 
\item we bet all of our wealth on each race (so $\sum_h B_h = 1$) and,
\item we re-invest all of our winnings, beginning with initial wealth $w_0$.  Because $\sum_h B_h P_h > 1$ for all $n$, 
a successful sequence of bets will cause our wealth, $w_n$, to grow exponentially with $n$. 
\end{itemize}

$$
w_n = \prod^t w_t = w_0 \prod_t \left\{\sum_h {\mathbf 1}_{ht} B_h \mathcal O_h \right\},
$$
where ${\mathbf 1}_{ht}$ is the random variable that is one if horse $h$ wins the $t^{\rm th}$ race and zero otherwise.  We 
then have 
$$
\log_{|\mathcal A|} w_n = \log_{|\mathcal A|} w_0 + \sum_t \sum_h {\mathbf 1}_{ht} \log_{|\mathcal A|} B_h \mathcal O_h.
$$
For $\mathbf B = (B_1, B_2, \ldots B_{|\mathcal A|})$ and $\mathbf P = (P_1, P_2, \ldots P_{|\mathcal A|})$, define the expected growth rate 
$G(\mathbf B, \mathbf P)$ as
\begin{align*}
G(\mathbf B, \mathbf P) &= {\mathbb E} \left[\sum_h {\mathbf 1}_{ht} \log_{|\mathcal A|} B_h \mathcal O_h\right]   \\
                                   &= \sum_h P_h \log_{|\mathcal A|} B_h \mathcal O_h   \\
\end{align*}
\cite{CoverAndThomas} shows that the optimal betting allocation is to choose $\mathbf B = \mathbf P$, 
even if the $P$'s are conditioned on the outcomes of the previous races.  It then shows that the increase in the growth 
rate due to the conditioning, at least if $S$ is ergodic, is the mutual information 
between the previous outcomes and the outcome of the current race\footnote{\cite{CoverAndThomas} also notes that the mutual information, and
thus the pragmatic information  
between the previous and current outcomes is an upper bound on this increase if the all-or-nothing outcome of the horse race is replaced by a 
vector of positive random variables representing daily price fluctuations.}.  Because the prior distribution in the pragmatic information formula is the
marginal distribution for $S[n]$, the mutual information between $S[{\ :}n]$ and $S[n]$ is the pragmatic information as well.\\

What \cite{CoverAndThomas} does not discuss is how the conditional distribution of $S[n]$, given $S[{\ :}n]$,  is determined.  Since this determination is, 
in fact, a prediction, we consider this problem from the point of view of predicting the continuation of an individual sequence.
We might expect to measure predictability by considering the number of possible choices for $S[n]$, given $S[{\ :}n]$, the number of such choices 
ranging from 1 (for perfectly predictable sequences) to $|\mathcal A|$ (for fully random sequences).  However, such measurements beg the question of 
how the prediction is being made.  As we will see, the difficulty of making these predictions, and thus the ability of a receiver to extract pragmatic information 
from $S[{\ :}n]$ can be characterized precisely by the Hausdorff dimension of the individual sequence $S$. (The remarkable result that an individual 
sequence {\it has} a Hausdorff dimension is demonstrated in \cite{LutzIndividualDimension}.)\\

\cite{Pushdown} observes that such predictions can not only be made with a Turing-equivalent computational element, 
but also with a less powerful computational element, such as a finite state machine or a pushdown automaton.  This observation leads to definitions 
of so-called finite state and pushdown dimensions, in which the Turing-equivalent computational element is replaced, respectively, by a finite state machine 
or a pushdown automaton. We conjecture that a corresponding dimension could be defined for the linear bounded automaton, though such a
definition has not, to our knowledge, appeared in the literature.  Given this conjecture, we would then have
a distinct definition of dimension for each of the classes of automata in the Chomsky hierarchy, to be denoted dim$_\mathfrak F(S)$, dim$_\mathfrak P(S)$, 
dim$_\mathfrak B(S)$, and dim$_\mathfrak T(S)$ for finite state dimension, pushdown dimension, linear bounded 
dimension, and Turing-equivalent dimension ({\it i.e.} the extended
Hausdorff dimension), respectively.  Since pushdown automata are strictly more powerful than finite state machines, 
linear bounded automata strictly more powerful than pushdown automata, and general Turing machines 
strictly more powerful than linear bounded automata, we would certainly have (with the conjecture properly flagged as such)
\begin{align}
{\rm dim}_\mathfrak F(S) \geq {\rm dim}_\mathfrak P(S) \stackrel{?}{\geq} {\rm dim}_\mathfrak B(S) \stackrel{?}{\geq} 
{\rm dim}_\mathfrak T(S).
\label{eq: DimensionInequality}
\end{align}
Below, we show that, except for the conjectured ones, the inequalities are, in some cases, strict.  We then show that such 
statements imply that the more 
powerful automata in the Chomsky hierarchy can extract qualitatively more pragmatic information from its input.\\

Given that ${\rm dim}_\mathfrak B(S)$ seems not to have even been defined yet, it should be no surprise that the entirety of the above claim must,
at present, remain a conjecture.  However, we establish the remaining pieces of the above claim and 
unpack their significance in the remainder of this section.  We start by characterizing those sequences from which pragmatic information can't possibly  
be extracted because they aren't predictable.  We do so by, first, considering the {\it Kolmogorov complexity},
$K(m)$ of the string $m$, defined as the length of the shortest computer program that can reproduce $m$.   
Evidently, by the pigeon hole principle, at least some sequences must have $K(m) \approx |m|$;
in other words, at least some sequences must be essentially incompressible.  The infinite sequence $S$ is then defined to be 
{\it Martin-L\"of random} \cite{MLRandomness}, if $K(S[\ :n]) \geq n - c$, for all sufficiently large $n$ and 
some constant, $c$, that does 
not depend on $n$.  Per \cite{LutzIndividualDimension}, Martin-L\"of random strings have Hausdorff dimension 1, the largest 
possible value, reflecting their maximal randomness.  Per \cite{Schnorr}, Martin-L\"of randomness has a number of other properties that 
coincide with our intuitive notions of randomness, such as the ability of 
Martin-L\"of random sequences to ``pass every algorithmically 
implementable statistical test of randomness''.\\ 

One way of formulating such tests is via a generalization of the idea of a martingale, familiar from probability theory.  In that theory, the term 
refers to a sequence of real valued random variables, $\ldots Z_{i-1}, Z_i, Z_{i+1}, \ldots $, 
such that $\mathbb E[|Z_i |] < \infty$ and
$$
\mathbb E\Big[Z_{i+1}| Z_i, Z_{i-1}, \ldots  \Big]  = Z_i,
$$
which is equivalent to
$$
\mathbb E \Big[Z_{i+1}- Z_i| Z_i, Z_{i-1}, \ldots  \Big]  = 0.
$$
Thus, such a sequence of random variables is a formalization of the cumulative winnings in a fair game, in which, on average, 
nothing is gained or lost.\\

We generalize this idea from real valued random variables to functions on strings by defining a martingale as any function, $w: \mathcal A^* 
\rightarrow \Re^+$, such that
$$
w(m) = \frac{1}{|\mathcal A|} \sum_{\alpha \in \mathcal A}w(m + \alpha),
$$
for all $m \in\mathcal A^*$ and the expression $m + \alpha$ is interpreted as the appending of the symbol $\alpha$ to the 
end of $m$.  As with the horse race, $w(m)$ can be interpreted as the total wealth 
generated by a sequence of bets on the successive symbols of $S$, given a unit initial investment, provided that, 
on the $i^{\rm th}$ bet,
\begin{itemize}
\item all capital is re-invested,
\item the fraction of the capital bet on a given symbol $\alpha^* \in \mathcal A$ is multiplied by $|\mathcal A|$ if $S[i]$ 
turns out to be $\alpha^*$, and 
\item the fraction of the capital bet on any other symbol in $\mathcal A$ is completely lost. 
\end{itemize}
We preserve the notion of a fair game in the sense that equal amounts of capital bet on each $\alpha \in
\mathcal A$ (the null bet) result in $w(S[{\ :}i]) = w(S[{\ :}i+1])$, regardless of the outcome of the $i^{\rm th}$ 
race.  This is because any losses arising from betting on the wrong symbols will be exactly recovered by
gains from betting on the right symbol.\\

Note that the above definition does not preclude martingales from systematically increasing as $i$ increases.  In fact, if 
the entire amount of existing capital, 
$w(S[{\ :}i])$, is bet correctly on the $i^{\rm th}$ 
race, $w(S[{\ :}i+1]) = |\mathcal A|\  w(S[{\ :}i])$, so any correct 
bet of most of the capital most of the time leads to exponentially growing capital.  
In contrast, 
if $S$ is random, then bets on the successive races will not. 
\\

In between the above two extremes are sequences that are somewhat predictable.  
For example, consider the set $C$ of all sequences defined on the trinary alphabet
$\mathcal A = \{\textrm{`0', `1', `2'}\}$ in which the `1' character doesn't appear, as in the original 
Cantor set, but is otherwise Martin-L\"of random.  A betting strategy that allocates half the existing capital to `0' and half to `2' is guaranteed to grow 
that capital by a factor of 1.5, either because the half of the capital allocated to `0' gets tripled or because the half of 
the capital allocated to `1' gets tripled.  The invested capital grows exponentially in both cases, albeit more slowly than 
for a fully predictable sequence, such as $\ldots \textrm{`0', `0', `0',} \ldots$, because {\it all} of the existing capital gets 
tripled with each bet in the latter case.\\

The relevant literature anthropomorphizes $w$ as the winnings of 
a ``gambler'' that ``succeeds'' on a given sequence, $S$, if
$$
\lim \sup_{i \rightarrow \infty} w_s(S[{\ :i]}) = \infty
$$
Formally, the gambler is a 
{\it betting function}, $B: \mathcal A^* \rightarrow \mathbb S_Q$, where  
$$
\mathbb S_Q = \left\{(x_1, x_2, \ldots, x_{|\mathcal A|}) \in [0, 1]^{|\mathcal A|} \  \middle| \ \sum_{i=1}^{|\mathcal A|} x_i = 1 \right\},
$$ 
{\it i.e.} the $|\mathcal A|$ dimensional unit simplex.  In the above, we assume that $w$ and thus $B$ are weakly 
Turing computable\footnote{A function is computable if a deterministic Turing 
machine can compute its value for any of its inputs and then halt.  A function is weakly computable as defined in \cite{Schnorr}
if its value is the limit of computable functions, each of which may be computed via a different Turing machine.} from 
$S[{\ :}i]$.\\

Following the (confusing) convention in the literature that it is $w$, rather than $B$ that succeeds, 
\cite{Schnorr} demonstrates that
\begin{theorem}
Let $W$ be the set of all possible martingales, let $w \in W$, let 
$$
\mathfrak S_w = \{S \in \mathbf \mathcal A^\infty | w \textrm{ succeeds on }S\},
$$
and let 
$$
\mathfrak S_W = \bigcup_{w \in W} \mathfrak S_w,
$$
then the set of sequences $S$ for which 
$\mathfrak S_W$ is empty are exactly the set of Martin-L\"of random sequences in $\mathcal A^\infty$.
\end{theorem}

The sensitive dependence of the growth of the gambling winnings on the growth rate suggests a precise measure of
predictability, which makes use of the following 
\begin{definition}
An $\mathbf{s-gale}$ is $w_s: \mathcal A^* \rightarrow \Re^+$, such that
$$
w_s(m) = \frac{1}{|\mathcal A|^s} \sum_{\alpha \in \mathcal A}w_s(m + \alpha),
$$
where $m$ is any string in $\mathcal A^*$ and $s \in [0, 1]$.  
\end{definition}
We restrict the value of $s$ to the unit interval because an $s$-gale with $s>1$ would have 
$$
\limsup_{i \rightarrow \infty} w_s(S[{\ :i]}) = \infty
$$
for the uninteresting null bet.  For $s<1$, 
$w_s$ can still be interpreted as the total capital generated by the sequence of bets 
on the successive symbols of $S$, but now with a ``tax'' of $|\mathcal A|^{1-s}$ on the capital after each bet.  The lower the 
value of $s$, the higher the tax.  
Therefore, we will only be able to find $w_s$ such that $\limsup_{i \rightarrow \infty} w_s(S[{\ :i]}) = \infty$ if 
and only if $S$ has some predictability, and the more predictable the sequence, the smaller we can make $s$ and still have 
$\limsup_{i \rightarrow \infty} w_s(S[{\ :i]}) = \infty$.\\

Let $\mathtt{SG}_s$ be the set of $s$-gales for a given $s$.  
The notion of an $s$-gale then permits us, following \cite{LutzIndividualDimension}, to define $\mathrm{dim}_\mathfrak T(S)$, the Hausdorff 
dimension\footnote{We use the subscript $\mathfrak T$ to emphasize the use of a general purpose Turing machine to
do the requisite computation, an assumption that we will relax shortly.}, 
of the individual sequence, $S \in \mathcal A^\infty$, as follows:
\begin{definition}
$$
\mathrm{dim}_\mathfrak T(S) = \inf \left\{s \in \Re^+ \middle|  
\exists w_s \in \mathtt{SG}_s
\mathrm{\ that\ succeeds\ on\ } S \right\}
$$
\end{definition}
Evidently, per the above theorem, the Hausdorff dimension of Martin-L\"of random sequences is 1 because the null bet 
for such sequences succeeds for all $s > 1$\\

Suppose we insist, instead, that $B$ is computed using a finite state machine.  The result is the so-called 
{\it finite state gambler} and a corresponding notion of randomness; 
namely, a sequence, $S$, is {\it finite state random} if no finite state gambler can succeed on it.  Per 
\cite{Pushdown}, {\it pushdown random} sequences can be defined similarly, and we see no reason
why {\it linear bounded automaton random} sequences can be similarly defined, as well.\\

If we define $w^\mathfrak F_s$ as an $s$-gale in which the betting is done with a finite state gambler and $\mathtt{FS}_s$ as 
the set of all such $s$-gales, we can define the 
{\it finite state dimension}, $\mathrm{dim}_\mathfrak F(S)$ of a sequence, $S \in \mathcal A^\infty$, as follows:
\begin{definition}
$$
\mathrm{dim}_\mathfrak F(S) = \inf \left\{s \in \Re^+ \middle|  
\exists w^\mathfrak F_s \in \mathtt{FS}_s \mathrm{\ that\ succeeds\ on\ } S \right\},
$$
\end{definition}
In other words, the finite state dimension of $S$ is the infimum of those real values  of $s$ for which 
$w^\mathfrak F_s$ succeeds on $S$.  The finite state 
dimension for finite state random sequences is also 1 because a finite state machine can always 
implement the null betthat succeeds for all $s > 1$.\\

\cite{Pushdown} proves the following
\begin{theorem}
$\mathrm{dim}_\mathfrak F(S) = 1$ if and only if $S$ is Borel normal,
\end{theorem}
where a sequence is Borel normal if any fixed string 
$F \in \mathcal A^*$ appears in $S$ as frequently as it would if the successive characters of $S$ are chosen with a uniform 
distribution over $\mathcal A$.  More formally, if 
$\#(F, S[{\ :i]})$ is the number of times that fixed string $F \in \mathcal A^*$ appears in the first $i$ symbols of $S$, 
a sequence $S \in \mathcal A^\infty$is Borel normal if,  for all $F \in \mathcal A^*$,
$$
\lim_{i \to \infty} \frac{\#(F, S)}{i} = \frac{1}{|\mathcal A|^{|F|}},
$$
and the limit exists. (See \cite{Pushdown} for an example of when it doesn't.)\\

It is widely conjectured (but, to our knowledge, yet to be proven) that the decimal expansions of many familiar irrational numbers, such as 
$\pi$, $e$, and $\sqrt{2}$, are Borel normal.  
However, \cite{Champernowne} proved that the so-called {\it Champernowne sequences}, $C_b$, are indeed Borel normal.  
These sequences are formed by simply 
concatenating the representations of successive natural numbers in the given base, $b$.  Thus, the base 10 version 
would be the sequence $C_{10} = \textrm{`1'\ +\ `2'\ +\ `3'\ }+ \ldots = \textrm{`1234567891011121314...'}$.  
Nevertheless, Champernowne
sequences of any base are clearly predictable by a general purpose Turing machine, from which it follows that $\mathrm{dim}_\mathfrak T(C_B) = 0$ for any base $B$.\\

\cite{Pushdown} defines the pushdown dimension, $\mathrm{dim}_\mathfrak P(S)$, similarly to 
$\mathrm{dim}_\mathfrak F(S)$, except that the finite state 
gambler in the latter is replaced by a {\it pushdown gambler} in the former.  The only difference between the pushdown 
gambler and the finite state gambler is
that the relevant computation is done by a betting-function-augmented, deterministic pushdown automaton, rather than a similarly augmented finite state machine.  
\cite{Pushdown} then demonstrates that $\mathrm{dim}_\mathfrak F(S) > \mathrm{dim}_\mathfrak P(S)$ for certain variations of Champernowne sequences.  It is tempting to 
conjecture that $\mathrm{dim}_\mathfrak P(C_b) = 1$ for any base $b$
, but we have not been able to prove this conjecture.  However, we have been able to identify another sequence, $\mathscr S$,
for which we have the following 

\begin{theorem}
There exists a sequence, $\mathscr S \in \mathcal A^{\infty}$, such that 
$\mathrm{dim}_\mathfrak P(\mathscr S)  = 1$ but $\mathrm{dim}_\mathfrak T(\mathscr S) = 0$.
\end{theorem}
\begin{proof}
We identify the successive symbols, $\mathscr S[n] \in \mathcal A$, of $\mathscr S$ as follows:
Consider the set $W^{\mathfrak P}$, the set of all martingales for which the betting function is a pushdown automaton.  
Since $W^{\mathfrak P}$ is countable, we give it a lexigraphic ordering, which enables us to identify, for all positive 
integers $n$, a unique 
martingale, $w_n$, which starts with capital $w_n(\mathscr S[n-1]) = 1$.  We then define the function 
$V:  \mathcal A^* \to \Re$ such that\footnote{Readers will note that the Python convention for string prefixes specifies that
$\mathscr S[\ :n]$ is the prefix that includes the symbols in $\mathscr S$ up to, but not including $\mathscr S[n]$.  When we
want to include $\mathscr S[n]$, we write $\mathscr S[\ :n+1]$.}
$$
V(\mathscr S[\ :n]) = \sum_{j=1}^{n-1} |\mathcal A|^{-j}w_j(\mathscr S[\ :n]).
$$
and the functions $E_\alpha:  \mathcal A^* \to \Re$ for all $\alpha \in \mathcal A$ 
$$
E(\alpha, \mathscr S[\ :n]) = \sum_{j=1}^{n-1} |\mathcal A|^{-j}w_j(\mathscr S[\ :n] + \alpha).
$$
We choose $\mathscr S[0] = \alpha_0$, the first symbol in the lexigraphical ordering of $\mathcal A$, and, for $n>0$,
we choose $\mathscr S[n]$ to be that symbol $\alpha_{\min}^{(n)}$ such that 
$$
E(\alpha_{\min}^{(n)}, \mathscr S[\ :n])= \min_{\alpha \in \mathcal A} E(\alpha, \mathscr S[\ :n]) 
$$
If the minimum of the $E$'s is not unique, we choose $\alpha_{\min}^{(n)}$ to be the $\alpha \in \mathcal A$ that 
appears first in $\mathcal A$'s lexigraphic ordering.   
We now claim that 
$$
V(\mathscr S[\ :n+1]) \leq V(\mathscr S[\ :n]) + |\mathcal A|^{-n}w_n(\mathscr S[\ :n+1]).
$$
The last term on the right of the above inequality is just the last term of $V(\mathscr S[\ :n+1])$, so we will have verified 
the claim if we can verify the inequality
$$
\sum_{j=1}^{n-1} |\mathcal A|^{-j}w_j(\mathscr S[\ :n+1]) \leq
\sum_{j=1}^{n-1} |\mathcal A|^{-j}w_j(\mathscr S[\ :n]).
$$
We use the fact that each $w_j$ is a martingale to write the right side of this latter inequality as
$$
\sum_{j=1}^{n-1} |\mathcal A|^{-j}w_j(\mathscr S[\ :n]) = 
\sum_{j=1}^{n-1} \sum_{\alpha \in \mathcal A}|\mathcal A|^{-j-1}w_j(\mathscr S[\ :n] + \alpha).
$$
The verification of the above claim is complete upon observing that, by our definition of $\alpha_{\min}^{(n)}$,
\begin{align*}
\sum_{j=1}^{n-1} \sum_{\alpha \in \mathcal A}|\mathcal A|^{-j-1}w_j(\mathscr S[\ :n] + \alpha) 
                            &\geq \sum_{j=1}^{n-1} |\mathcal A|^{-j}w_j(\mathscr S[\ :n] + \alpha_{\min}^{(n)}) \\
                            &=\sum_{j=1}^{n-1} |\mathcal A|^{-j}w_j(\mathscr S[\ :n+1])
\end{align*}
Because the newly introduced $n^{\rm th}$ martingale $w_n$ can only bet a fraction of its initial capital of 1, its 
maximum possible value after the step after it is introduced is $|\mathcal A|$. Thus, 
$$
|\mathcal A|^{-n} w_n (\mathscr S[\ :n+1] ) \leq |\mathcal A|^{1-n},
$$
and iterating the inequality in the above claim, we conclude that, for all $n$,
$$
V(\mathscr S[\ :n+1]) \leq V(\mathscr S[0]) + \sum_{j=1}^n |\mathcal A|^{1-j } < \infty
$$
Since $V(\mathscr S[\ :n])$ is bounded for all $n$, every individual term in the sum that defines it is bounded, {\it i.e.}
$|\mathcal A|^{-k}w_k(\mathscr S[\ :n]) < c_k$ for all $n$ and some real constant $c_k$ that does not depend on $n$.  
In other words, no $w^{\mathfrak P} \in W^{\mathfrak P}$ can succeed on $\mathscr S$, so
$\mathrm{dim}_\mathfrak P(\mathscr S)  = 1$.  However, every symbol in the above construction has been identified by
a deterministic algorithm; hence, each symbol in $\mathscr S$ is fully predictable by a Turing machine.  It follows that 
$\mathrm{dim}_\mathfrak T(\mathscr S) = 0$.
\end{proof}

With the above theorems, we have indeed verified that the inequalities in \eqref{eq: DimensionInequality}, with the exception
of those involving the hypothetical bounded linear Turing dimension, ${\rm dim}_\mathfrak B(S)$, can sometimes be strict.
\\

Next, in order to make contact with notions of pragmatic information, we introduce probabilities.  We postulate that, {\it a priori}, $S[i]$ can be any of the symbols in $\mathcal A$.  In the previous section's 
language of symmetries, we might call this the ``symmetry of ignorance'', where every possibility is as likely as any other, 
because we have nothing that tells us otherwise.  Thus, $\mathbf q = \mathbf 1/|\mathcal A|$.
In processing $S[{\ :n}]$, the automaton 
$\mathfrak A \in \left\{\mathfrak F, \mathfrak P, \mathfrak T\right\}$  
has partially broken the symmetry of ignorance by reducing the number of possibilities for $S[n]$ from 
$|\mathcal A|$ to roughly $|\mathcal A|^D$, where $D<1$ is the relevant dimension of $S$.  We therefore postulate that 
$m=S[{\ :n}]$ and $p_{\alpha'|m} = 1/|\mathcal A|^D$ if symbol $\alpha' \in \mathcal A' \subseteq \mathcal A$ 
remains a possibility for 
$S[n]$ and zero otherwise.  Since the $D$-gale $w_D$ is the 
decision maker of previous sections, with input $m = S[{\ :n}] \in \mathcal A^n$ and output $\alpha \in \mathcal A$, we 
conclude that the pragmatic information per bet is
\begin{align*}
\Phi_\Delta (\mathcal S; \mathcal A) &= \lim_{n \to \infty} \Phi_\Delta (\mathcal A^n; \mathcal A) \\
                                                     &= \lim_{n \to \infty} \sum_{\alpha \in \mathcal A, m \in \mathcal A^n} \varphi_m  p_{\alpha|m} 
						                           \log \left(\frac{p_{\alpha|m}}{q_\alpha}\right)  \\
                                                     &= \lim_{n \to \infty} \sum_{\alpha' \in \mathcal A', m \in \mathcal A^n}
                                \frac{\varphi_m}{|\mathcal A|^D} \log \left[\frac{|\mathcal A|^{-D}}{|\mathcal A|^{-1}}\right]  \\
      &= (1-D) \log |\mathcal A|.\\
\end{align*}
Results for gamblers of each of the three automata for which the relevant dimension is defined are identical, 
except, as we have seen, the actual numerical values of $D$ can sometimes 
be different for these different automata classes.  In such cases, the more powerful 
class of automata, $\mathfrak A^*$, will be able to extract 
$$
\log |\mathcal A| \left[\rm{dim}^{\mathfrak A^*}(S) - \rm{dim}^\mathfrak A(S)\right] > 0
$$
more bits of pragmatic information per bet than the less powerful class of automata,  $\mathfrak A$.  Evidently, this result 
will be also be true for minor variations of the above probabilities.  The conclusions of this subsection can therefore be 
summarized by the following
\begin{theorem}
There are choices of the above probabilities and sets of symbol sequences, $\mathcal S$, for which 
$\Phi_\mathfrak F (\mathcal S; \mathcal A) < \Phi_\mathfrak P (\mathcal S; \mathcal A)$.  Similarly, there are choices of 
probabilities and sets of symbol sequences, $\mathcal S'$,  for which
$\Phi_\mathfrak P (\mathcal S'; \mathcal A) < \Phi_\mathfrak T (\mathcal S'; \mathcal A)$.
\end{theorem}
\begin{theorem}
There are sets of symbol sequences, $\mathcal S$, for which 
$\rm{dim}^{\mathfrak F}(S) = 1$ for $S \in \mathcal S$
but $\rm{dim}^{\mathfrak T}(S) = 0$ for $S \in \mathcal S$.  Similarly, 
there are sets of symbol sequences, $\mathcal S'$, 
for which $\rm{dim}^{\mathfrak P}(S') = 1$ for $S \in \mathcal S'$
but $\rm{dim}^{\mathfrak T}(S') = 0$ for $S' \in \mathcal S$.  
\end{theorem}
Thus, there are deterministic symbol sequences from which no pragmatic information can be extracted 
by a finite state machine because these sequences appear random to them, but pragmatic information 
can be extracted by a pushdown automaton.  Similarly, there are deterministic symbol sequences from 
which no pragmatic information can be extracted 
by a pushdown automaton because these sequences appear random to them, but pragmatic information 
can be extracted by a Turing machine. (While we conjecture that the same is true for linear bounded 
automaton, we must await the formalization of the requisite definitions to be sure.)  
\subsection{A Case in which Enhanced Computational Ability Doesn't Increase Pragmatic Information but Does Increase Pragmatically Useful Information }
This section proves the following
\begin{theorem} Enhanced computational abilities can increase pragmatically useful information, even if these
enhanced abilities don't increase the messages' pragmatic information.  
\end{theorem}
\begin{proof} Suppose that $\Delta$ processes messages via the finite state machine $\mathfrak F$, which outputs $\omega_1$ 
if $m \in \mathcal L \subset \mathcal M$; otherwise, $\mathfrak F$ outputs $\omega_0$ (As per the above, $\mathcal L$ is 
necessarily a regular language.).
If we assume further that the respective prior probabilities of $\omega_0$ and $\omega_1$ are $q_0$ and $q_1$, that 
$\mathbb P \left\{ m \in \mathcal L \right\}$ is the probability that $m \in \mathcal L$,
and that $\mathfrak F$ can perfectly determine whether a given $m$ belongs in $\mathcal L$ or not ({\it i.e.} $\mathcal M$ is 
 pragmatically definitive), then 
\begin{align*}
\Phi_{\mathfrak F} ({\mathcal M}; \Omega) &= \sum_{m \notin \mathcal L} \varphi_m  p_{0|m} \log \left( \frac{p_{0|m}}{q_0} \right)  + \sum_{m \in \mathcal L} \varphi_m  p_{1|m} \log \left( \frac{p_{1|m}}{q_1} \right) \cr
                         &=\sum_{m \notin \mathcal L} \varphi_m   \log \left( \frac{1}{q_0} \right)  + \sum_{m \in \mathcal L} \varphi_m   \log \left( \frac{1}{q_1} \right) \cr
                         &=-\Big[1-\mathbb P\left\{ m \in \mathcal L \right\} \Big]\log q_0 -\mathbb P\left\{ m \in \mathcal L \right\} \log q_1.
\end{align*}
Finally, suppose we agree that all messages $m \in \mathcal M$ are pragmatically useful when $\mathfrak F$ classifies them properly.  If $m \in \mathcal L$ is a proper classification, $\Phi_{\mathfrak F} ({\mathcal U}; \Omega) = \Phi_{\mathfrak F} ({\mathcal M}; \Omega)$.  This will no longer be true if we make more stringent computational demands, as we now show. \\

To make the above demonstration, we need the {\it pumping lemma for regular languages} \cite{Sipser}, which states that there is an integer $p \geq 1$, depending only on $\mathcal L$, such that 
\begin{itemize}
\item every message $m \in \mathcal L$ with length $|m| \geq p$ can be split into three substrings $x, y$ and $z$, such that $m$ is their concatenation, {\it i.e.} $m = xyz$, 
\item $|y| \geq  1$, 
\item  $|xy| \leq p$, and 
\item  every string of the form $m = xy^nz$ is in $\mathcal L$, with $y^n$ defined as the $n$-fold concatenation of $y$ with itself.
\end{itemize}
Consider now a different language, $\mathcal L'$, with the same messages as those in $\mathcal L$, except that strings of the form $m' = xy^nz$ are {\it not} in $\mathcal L'$ if $n$ is prime.  We now suppose that the only pragmatically useful messages
are those in $\mathcal L'$.  Because no finite state machine $\mathfrak F$ can determine whether a given number is prime \cite{AHU}, no $\mathfrak F$ can determine that a given input message is in $\mathcal L$, but not in $\mathcal L'$.  The inability to make this determination does not affect the value of $\Phi_{\mathfrak F} ({\mathcal M}; \Omega)$.  However, it decreases the value of $\Phi_{\mathfrak F} ({\mathcal U}; \Omega)$ for any assignment of probabilities to the messages in $\mathcal M$ for which $\mathcal L \setminus \mathcal L'$ has positive probability.  It follows that $\Phi_{\mathfrak F} ({\mathcal U}; \Omega) < \Phi_{\mathfrak F} ({\mathcal M}; \Omega)$ in such cases, because the mis-classification of messages in $\mathcal L \setminus \mathcal L'$ does not affect $\Phi_{\mathfrak F} ({\mathcal M}; \Omega)$.
\end{proof}

\subsection{Why Pragmatic Information Need Not Be Computable}

First, a bit of background:  Consider a function $f: \mathcal X_f \rightarrow \mathcal Y_f$, where
$\mathcal X_f \subset \mathcal A^*$ and $\mathcal Y_f \subseteq \mathcal A^*$ (The reason for the important restriction that $\mathcal X_f$ is a strict subset of $\mathcal A^*$ will become clear in a moment.).  Suppose $f$ is computable, meaning that there is a Turing machine, $\mathfrak T_f$, that always halts with $f(x) = y$ on its tape if $x$ is on its tape initially,
for all $x \in \mathcal X_f$.  The existence of universal Turing machines implies that $\mathfrak T_f$ can be simulated by some other Turing machine $\mathfrak T_{\mathfrak U}$ in the following sense:  For any computable $f$, there is some $p \in \mathcal A^*$, depending on $f$, such that, when $p$ is prefixed to $x$ on $\mathfrak T_{\mathfrak U}$'s input tape, outputs $f(x) = y$ and halts, for all $x \in \mathcal X_f$ (Informally, $p$ can be thought of as a computer program for which $x$ is the input and $f(x) = y$ is the output; more formally, there is a function $\mathfrak U: \mathcal A^* \rightarrow \mathcal A^*$,
computable on the subset, 
$\mathbb D_\mathfrak U$, of strings for which $\mathfrak T_{\mathfrak U}$ halts and a string $p \in \mathcal A^*$ such that $\mathfrak U(p x) = f(x) = y$ for any computable $f: \mathcal X_f \rightarrow \mathcal Y_f$.  Note that implicit in this construction is the requirement that such strings $p x$ are prefix free, in other words, that none of these strings is a prefix of any of the others.  A necessary and sufficient condition for a countably infinite set of strings $s_1, s_2, \ldots $ to be prefix free is the Kraft inequality, namely that
$$
\sum_{k=1}^\infty |\mathcal A|^{-|s_k|} \leq 1.
$$
Since the set of all possible prefix free strings
contains some input strings to $\mathfrak U$ for which $\mathfrak U(m)$ is undefined because $
\mathfrak T_{\mathfrak U}$ doesn't halt, the Kraft inequality is strict for $m \in \mathbb D_{\mathfrak U}$, {\it i. e.}
$$
\mathbf \Omega_{\mathfrak U} = \sum_{m \in \mathbb D_{\mathfrak U}} |\mathcal A|^{-|m|} < 1.
$$
$\mathbf \Omega_{\mathfrak U}$ is sometimes called Chaitin's constant (though its actual value can differ radically depending on the details of how $m$ is encoded in $\mathfrak U$) \cite{ChaitinConstant}.  
$\mathbf \Omega_{\mathfrak U}$ can be interpreted as the probablilty that a 
universal Turing machine of a given construction will halt when its tape is initialized with the string $m \in \mathcal A^*$
when $m$ is chosen with probability $|\mathcal A|^{-|m|}$ (Note 
that these probabilities are well defined because $\mathbb D_{\mathfrak U}$ is prefix free.).\\

$\mathbf \Omega_{\mathfrak U}$ is not Turing computable because computing it it would require a knowledge of which 
strings $m$ are in $\mathbb D_{\mathfrak U}$, {\it i.e.} whether $\mathfrak U$ halts on such $m$'s.  For this reason, 
$\mathbf \Omega_{\mathfrak U}$, was one 
of the first specific numbers to be proven algorithmically random.
\\

We can now state and prove the following 
\begin{theorem} 
$\Phi_{\Delta} ({\mathcal M}; \Omega)$ can be Turing uncomputable if $\mathfrak A_\Delta$, $\Delta$'s input processor, is the equivalent of a general purpose, deterministic Turing machine.
\end{theorem}
\begin{proof}
Suppose that the decision that $\Delta$ is obliged to make is whether $\mathfrak A_\Delta$ will halt for a given message $m$.  We have the two states: $\omega_0 = \mathfrak A_\Delta$ halts on $m$ and $\omega_1 = \mathfrak A_\Delta$ doesn't halt on $m$, so the two prior probabilities are, respectively, $q_0 = \mathbf \Omega_{\mathfrak U}$ and $q_1 = 1 -\mathbf \Omega_{\mathfrak U}$.  We then present 
$\Delta$ with the single message $m^*$, which we supply with probability one.  $p_{0|m}$ is then the probability that  $\mathfrak A_\Delta$ halts on $m^*$, for the particular message $m^*$, and 
$1 -p_{0|m} = p_{1|m}$ is the complementary probability.  Since  $\mathfrak A_\Delta$ is deterministic, $p_{0|m} \in \{0, 1\}$, but, in general, we can't tell which because of the Halting Theorem.  
It follows that we cannot compute 
$$
\Phi_{\Delta} ({\mathcal M}; \Omega) = p_{0|m} \log\left(\frac{p_{0|m}}{\mathbf \Omega_{\mathfrak U}}\right) + p_{1|m} \log\left(\frac{p_{1|m}}{1-\mathbf \Omega_{\mathfrak U}}\right), 
$$
because we can't compute {\it any} of the terms in the right hand side of the above.
\end{proof}

In other words, it may simply not be possible, in the general case, to assess the impact that a given message has on a given receiver.

\section{Towards a ``Channel Capacity'' for Pragmatic Information}
As noted above, traditional information theory is primarily concerned with the accurate transmission of symbols in the 
presence of interference that may garble the symbols in transit.  By treating this interference as a random process 
that can unpredictably substitute some of the transmitted symbols with others, the traditional theory derives its famous ``noisy
coding theorem''.  This theorem establishes the existence of a so-called channel capacity, 
a well-defined maximum rate at which a message can be transmitted over a given communications channel with arbitrarily high fidelity.  
In this section of the paper, we consider the sense in which there is an analogous ``pragmatic channel capacity'', {\it i.e.} a 
maximum rate at which {\it meaning} can be transmitted.\footnote{Note that our notion of pragmatic noise includes, not only 
transmission errors, but also misinterpretations of correctly transmitted symbols.}  
The non-zero time required to transmit a message in classical information theory is roughly equivalent to the non-zero time it 
takes to parse a message (and thus solve the membership problem) in our framework because the pragmatic information 
conveyed by the message is always zero until the message is parsed.  If $\tau_{\mathcal M}$ is this parsing time, 
it is tempting to interpret the ratio $\mathbb R = \Phi_{\Delta} (\mathcal U; \Omega)/\tau_{\mathcal M}$ as the rate at which 
pragmatically useful information, and thus meaning is transmitted.\footnote{The Victor Hugo anecdote related above 
suggests that considerable meaning can 
be expressed very concisely. Yet such examples do not negate the notion of a pragmatic channel capacity; instead, they simply 
demonstrate that this capacity can be very high for a suitably prepared receiver.}
However, if $\Omega$ is a finite set, as we have been assuming all along, $\mathbb R \to 0$ as 
$\max_{m \in \mathcal M} |m| \to \infty$.\\

Nevertheless, it is illuminating to consider how $\tau_{\mathcal M}$ depends on the automaton doing the parsing. 
If we assume that the message is parsed by the least powerful automaton capable of extracting all of the available pragmatic information, we get 
parsing time bounds that are based on the time complexity of parsing messages in the corresponding level in the Chomsky Hierarchy. In each case, the relevant input to the time complexity estimate is 
$|m^+|$, with $m^+ \in \mathcal U^+$ and $\mathcal U^+$ symbolizing, respectively, a single message $m$ augmented by whatever internal data is required for its parsing, and 
the ensemble of all pragmatically useful messages, each similarly augmented.\\

We note the dramatic differences in parsing times for languages at different levels of the hierarchy.\\

{\it Regular Languages}\\
  
As we see from Table 1, words in regular languages can be parsed by finite state machines, such as $\mathfrak F$, as described in the Appendix.  The time complexity required for a finite state machine 
with $\mathcal S$ states to process an input of length  $|m^+|$ and thus recognize $m^+ = d + m$ as a member of a 
given regular language (or not) is $\mathcal O(|\mathcal S| |m^+|)$.  This is because
each symbol in $m$ can induce a transition to one of a fixed, finite number
of states upon input and there can be at most $|S|$ additional $\Lambda$ ({\it i.e.} spontaneous) transitions per input symbol.
Note that this is true whether or not $m^+ \in \mathcal U^+$.  Note also that Theorem 4.2.1 shows that 
$\Phi_{\Delta} (\mathcal U^+; \Omega)$ is bounded, independent of $|m^+|$, thus validating the above claim that 
$\mathbb R \to 0$ as $|m^+|$ gets large.
\\

{\it Context Free Languages}\\

As we see from Table 1, recognizing a word in a given context free langugage requires parsing the input to a non-deterministic pushdown automaton, such as $\mathfrak P$, as described in Appendix B.  
Doing this involves the multiplication of boolean matricies of dimension $|m^+|$ (See, for example, \cite{ValiantParsing}).  Multiplying such 
 matrices would appear to require a running time of $\mathcal O(|m^+|^3)$; however, \cite{ValiantParsing}) notes that the asymptotic time complexity of any matrix multiplication can be reduced by cleverly dividing the
 matricies into blocks before multiplying them.  \cite{ValiantParsing} was written in 1974, when the best asymptotic time complexity for such multiplications was $\mathcal O(|m^+|^{\log_2 7}) \approx 
 \mathcal O(|m^+|^{2.81})$.  Meanwhile, these algorithms have steadily improved; as of this writing, the best 
 asymptotic time complexity for such multiplications is approximately $\mathcal O(|m^+|^{2.37})$.  Whether further improvement is possible is a major open question in computer science.  Nevertheless, any 
 algorithm for multiplying two $|m^+| \times |m^+|$ matricies must process all of the entries in the given matricies, so recognizing a word of length $|m^+|$ must necessarily require a running time of 
 at least $\mathcal O(|m^+|^2)$.\\

{\it Context Sensitive Languages} \\

Recognition of words in a context sensitive language is known to be $\mathtt{PSPACE}$ complete \cite{PSPACE}.  Solutions to such $\mathtt{PSPACE}$ complete problems require 
a memory size that grows at most as a polynomial in the length of the input, and every
other problem that can be solved using such a bound can be transformed to the given problem with a computation 
requiring at most a number of steps that is bounded by a polynomial in the length of the input.  Evidently, the solution to 
$\mathtt{PSPACE}$ complete problems requires at least the computational resources required 
for solving problems in $\mathtt{NP}$ \cite{PSPACE}, 
which are problems whose solutions can be verified, albeit not necessarily
solved, in a number of steps that can also grow no faster than a polynomial in $|m^+|$,  It is 
widely believed that $\mathtt{PSPACE}$ complete problems are, in general, more computationally demanding than $\mathtt{NP}$ complete problems.  
Whether $\mathtt{NP}$ problems can actually be  
solved in a number of steps that can also grow no faster than a polynomial in $|m^+|$ is the content of the famous $\mathtt{NP} \neq \mathtt{P}$ conjecture.  If this conjecture proves to be true, as most experts
believe, then the recognition of words in a context sensitive language,
and thus the running time required to parse the pragmatic information in $m$ grows faster than {\it any} polynomial in $|m^+|$.\\

{\it General Recursively Enumerable Languages} ($\mathfrak T$)\\

Because of the halting problem, the parsing time required to parse an arbitrary word in a general recursively enumerable language is, in general, 
unknowable in advance, implying that it could be arbitrarily large.  This also implies that there can be an arbitrarily long 
waiting time to determine whether or not the receiver can understand the message at all.\\

In summary, it is only at the lowest level of the hierarchy, that of regular grammars, that 
conforming messages can be processed in real time, since messages conforming to higher levels in the hierarchy cannot, in 
general, be processed in $\mathcal O(|m^+|)$ time.  Thus, a pragmatic information {\it rate} can only be defined when $m$ is 
a word in a regular language.  This difficulty in processing messages 
in their full complexity has led to various heuristics to process special cases, as we illustrate in our discussion of the efficient 
market hypothesis in the next section.  We also note the fundamental tradeoff implied by the results of this and the preceeding 
section:  the additional computational power gained as we ascend the Chomsky hierarchy comes at the cost of dramatically 
increased processing times.

\section{Implications for the Efficient Market Hypothesis of Financial Economics}

The {\it efficient market hypothesis} or EMH, in the words of one of its leading exponents, Nobel Laureate Eugene Fama \cite{Fama70, Fama91}, makes the claim 
that ``securit[ies] prices reflect all available information''.  Discussions of this hypothesis usually include the additional claim that 
securities prices reflect their ``fundamental value'' \cite{NoFreeLunch}, which is the discounted sum of the expectation of future cashflows accruing 
to owners of these securities and are thus in market clearing equilbrium with all other prices in the economy.  Fama identified three 
versions of the EMH, for three different possible sets of ``available information'', namely
\begin{itemize}
\item the weak form of the hypothesis, which asserts that all information in previous securities prices are fully reflected in current prices,
\item the semi-strong form of the hypothesis, which asserts that, not just previous prices, but {\it all public information}, 
        from sun spot cycles to the details of geopolitics, are fully reflected in current prices, and
\item the strong form of the hypothesis, which asserts that, not just all public information, but also 
{\it all insider information}, such as yet-to-be announced quarterly earnings information, 
        is fully reflected in current prices.  
\end{itemize}

Because we will be interested in the information in the price series itself, we will focus on weak 
form efficiency, which is most directly addressed by the celebrated {\it tâtonnement}\footnote{The French word 
{\it tâtonnement} is often translated as ``groping'' or ``trial and error''.}
 process of price formation \cite{groping}.  This process posits a hypothetical auctioneer who calls out price suggestions, 
aggregates supply and demand at each suggestion, and revises price suggestions accordingly until
supply matches demand.  Only these market clearing, equilibrium prices are the ones at which actual transactions take place, and it is the information in these that are supposed 
to be reflected in current prices.\\

The {\it tâtonnement} process differs significantly from the price discovery process of actual securities exchanges, such as the New York Stock Exchange (NYSE).  
Such exchanges usually function as so-called double auction markets.  As we will see in the more 
complete discussion below, the prices quoted in such markets are binding, even though there is no guarantee that these prices are market clearing, 
equilibrium prices.  At best, therefore, price information is only an estimate of where the true market equilibrium lies, suggesting that the true equilbrium price should 
be treated as a random variable, much as \cite{EfficientImpossibility} proposed.\\

It is at this point that the significance of the pragmatic information paradigm becomes evident.  Market participants must 
make a decision, in this case the decision to buy, sell, or hold a particular security at a particular 
price, based on their estimates of a random variable, in this case the true equilibrium price.  In this context, price information 
becomes relevant only to the degree that it allows the market participant to update their estimate of this random variable.  
We therefore recall the proposal in \cite{Weinberger24} 
that the EMH be recast as the assertion that a market is efficient 
{\it to a given market participant} if none of the information available to that market 
participant is pragmatically useful, as defined above.  One immediate benefit of this recasting is that it resolves 
a major problem with the existing theory of efficient markets:  although the same market information is available to all investors (absent insider information), 
the outstanding performance of a few investors makes a mockery of market averages.  Consider, for example, 
the Renaissance Medallion Fund, started by the legendary investor Jim Simons, that generated average annual returns of 
39.9\% per year from 1989 to 2022, {\it net of fees}, with only one of those years (2022) 
generating less than 20\%.  We can
only conclude from such a track record that market information is more ``available'' to some investors than others.  
\\

Previous sections of this paper suggest that the pragmatic utility of market information depends on
the computational resources of these market participants.  The purpose of this section is to elaborate on this theme for 
financial markets.  We therefore distinguish between markets that are  
\begin{itemize}
\item {\it finite state efficient} if no finite state machine can succeed in the sense of Section 4, above on the relevant information set, which would require that market returns are finite state random,
\item {\it pushdown efficient} if no pushdown automaton can succeed on the relevant information set, which would require that market returns are pushdown automaton random,
\item {\it linearly bounded efficient} if no linear bounded automaton can succeed on the relevant information set, which would require that market returns are linearly bounded automaton random, and
\item {\it Turing efficient} if no general purpose Turing machine can succeed on the relevant information set, which would require that market returns are Martin-L\"of Random.
\end{itemize}
Recall 
that each of the automata classes in the Chomsky hierarchy are strictly less powerful than automata in the classes above
them.  Hence, markets efficient with respect to a given automata class are subsets of markets that are efficient 
with respect to the less powerful automata below them in the hierarchy.  The above discussion shows that this subset relation 
is sometimes strict.\\

Our claim that computational issues are relevant to the weak form of market efficiency was previously considered by 
\cite{Maymin} and \cite{ComputationalView}.  Indeed, our notions of finite state efficiency, pushdown efficiency,  
etc. are simply elaborations of the notion in \cite{ComputationalView} of efficiency with respect to a resource. 
However, we differ from this previous work in two ways:
\begin{enumerate}
\item \cite{Maymin} and \cite{ComputationalView} treat the problem of ``beating the market'' as merely a recognition problem, namely that of
identifying appropriate entry points for trades.  In fact, trading profits are not realized until a trade is ``closed out'' 
by selling long positions and/or covering short positions.  This observation suggests that the relevant problem in the 
theory of formal languages 
is not the recognition problem, but the reachability problem, which asks, given the current state of the market, will a given final 
state be achieved?  
\item\cite{Maymin} asserts in its title that``Markets are efficient if and only if $\mathtt{P} = \mathtt{NP}$''.  However, that paper fails to consider the importance of the limit order book, as we discuss in detail below.  As part of discussion, we prove formally that any attempt to use ``all available [market] information'' is intrinsically $\mathtt{PSPACE}$ complete, which, as we noted in the previous section, is
an even greater level of intractability than 
$\mathtt{NP\ }$ completeness.  Thus, in practice, market participants may be stymied by the sheer 
volume of requisite computation and cannot extract pragmatic information from previous prices because of what we might 
term {\it computational efficiency}.  
\end{enumerate}
We present our discussion of computational efficiency in the remainder of 
this section, prefacing it with a brief account of market mechanics and a stylized account of how 
pragmatic information characterizes a market's approach to computational efficiency. 
We focus on the approach to finite state efficiency, rather than the linear bounded efficiency which we believe to be the 
true limiting case.  We justify this focus by our belief that the simpler finite state case illustrates the more 
general case and because the approach to finite state efficiency is closer to current market conditions.
We then provide our proof that market dynamics are, in the general case, $\mathtt{PSPACE}$ complete.  
While a linear bounded
automaton can, in principle, solve the $\mathtt{PSPACE}$ complete reachability problem of determining 
whether a given future market state will be achieved, the required compute is conjectured to scale 
exponentially with the length of problem's input.  Thus, computational efficiency obtains in 
practice because the requisite input is so large.
\\

\subsection{A Simplified Description of Market Mechanics and of the Computation of Market Returns}
In the standard continuous double auction market 
\cite{SFIonDoubleAuction} typified by the New York Stock Exchange\footnote{
While the description of market microstructure presented here only applies in detail to organized exchanges,
the exchange merely formalizes the supply and demand dynamics of securities traded outside of exchanges.}, 
software known as a matching engine brings together buyers and sellers.  Some of these market 
participants stand ready to buy pre-specified 
amounts of a security at a price less than or equal to a given bid, and others stand ready to sell pre-specified 
amounts of the security at a price greater or equal to a given offer.  A record of these so-called {\it limit orders} is maintained by the exchange's matching engine 
in a {\it limit order book} or LOB.  Alternatively,  a buyer or seller can place an order to buy or sell a specific amount of a security 
``at the market'' (buy at best offer or sell at best bid).  Absent such market orders, the bids in the LOB will be strictly less than the offers; at such times, absent the modification
or cancellation of an existing limit order, nothing happens and the LOB remains unchanged.   Thus, a transaction takes place only when 
\begin{itemize}
\item  a buyer or seller places a market order, 
\item a buyer or seller places a {\it marketable limit order} for a specific amount of the security that meets existing bids or offers
\end{itemize}
In such cases, the filled limit orders are removed from the LOB.
This simplified description of market mechanics is almost sufficient for our subsequent discussion; the 
one addition required is the so-called ``immediate or cancel'' (IOC) limit order.  
The so called resting limit orders of the previous paragraph remain in force until either they are filled, 
possibly via a series of partial fills, 
or they 
are explicitly modified or cancelled.  In contrast to resting limit orders, IOC orders are 
cancelled if they cannot be executed in full immediately upon receipt.\\

This market structure implies an intrinsic uncertaintly as to the exact price 
of the security most of the time, partly because, most of the time, the best bid is
strictly lower than the best offer.  Will it be a buyer or will it be a seller that will prove to be the more motivated to ``cross the
bid/offer spread''?  Only time will tell.  
And even
when there is a transaction, there is an agreement as to the value of a specific amount of the security {\it at that instant}.  
The price for a different amount of the security may well be different because, for example, a large market 
buy order  could be filled by several different sellers offering smaller amounts of the security at different prices.  Additional 
uncertainty is introduced by ``bid-offer bounce'' or a series of transactions 
that alternate between eager buyers trading at the unchanging best offer and eager sellers trading at the unchanging best bid. 
Then there are other market practices that obscure large limit orders, such as ``dark pools'' 
(outside-of-exchange order matching) and ``iceberg orders'', which are a series of orders that are automatically repeated as 
soon as the previous one is filled, thus concealing their true aggregate size.  Typically, the LOB is visible to market participants, 
but dark pools and iceberg orders are, in their totality,
invisible by definition.  Yet such orders account for a significant fraction of trading activity.  For example, 
{\it The Wall Street Journal} \cite{WSJ} reported that ``On some days [in 2021], 
more than half of shares changing hands in U.S. are traded outside public stock exchanges''.  Also, 
\cite{Iceberg} reported that, consistent with previous estimates, iceberg orders comprised more than 25\% of the volume of the E-mini S\&P 500 
September 2019 Futures Contract traded on the Chicago Mercantile Exchange on June 19, 2019.  \\

The difficulty in making any precise claim about the ill-defined notion of``securities prices'' confirms our view of true 
equilibrium prices as random variables.  To some degree, such difficulties 
are mitigated by averaging the transaction prices, perhaps weighted by transaction amounts, over some short time period, such as 1 or 5 
minutes.  Nevertheless, the failure of \cite{Maymin} to take market structure into consideration remains, and, as
we will see, including it can lead to significantly different conclusions.

\subsection{A Sketch of How Pragmatic Information Characterizes the Approach to Finite State Efficiency}
We base our analysis on a variant of a model due to \cite{ComputationalView}.  In this model, zero-knowledge market 
participants trade a single security, oblivious to a deterministic, periodic pattern in the dollar prices, $P_t$, of the security 
at successive time periods.  Suppose, 
in particular, that the period is reflected in daily closing prices, such as the following periodic sequence
$$
6, 5, 6, 5, 6, 5, 6, 5, 6, 7, \quad 6, 5, 6, 5, 6, 5, 6, 5, 6, 7, \quad 6, 5, 6, 5, 6, 5, 6, 5, 6, 7, \ldots,
$$
where extra spaces are inserted in the above to emphasize the periodicity\footnote{\cite{ComputationalView} notes that this model is not as contrived as it might seem, as there are indeed many sources of periodic predictability in the markets.  To the list of these provided by \cite{ComputationalView}, we add predictabile short term interest rate fluctuations near the beginning of each year and the periodic expiry of ``front month'' futures contracts.}.  After a certain number of such periodic fluctuations, 
a given market participant, whom we denote as $\Delta_1$,  
realizes that profits can often be made by observing the most recent recent price difference, 
$D_t = P_t - P_{t-1}$ and making $\mathtt{BUY/SELL}$ decisions accordingly\footnote{Strictly speaking, $\Delta_1$ cannot simultaneously observe a closing price and execute trades at that price.  Nevertheless,
prices of heavily traded securities rarely fluctuate substantially during the near instantaneous time required to execute a market order, or even reverse a position; thus, the result of observing prices and implementing 
the proposed strategy immediately before the close would closely approximate the results described here.}.  We might further suppose that the probability that $\Delta_1$ makes such decisions  
increases as $D_t$ increases in magnitude, perhaps according to the formula
$$
{\rm Pr}\left\{\mathtt{BUY\ |\ D_t}\right\} = \frac{1}{1 + e^{\kappa D_t}} 
                                                              = 1 - {\rm Pr}\left\{\mathtt{SELL\ |\ D_t}\right\},
$$
with the probabilities taken to be independent at each daily time step, and with $\kappa$ viewed as a measure of how useful 
$D_t$ is in predicting the next price.  This contrasts with the zero-knowledge market participants, for whom 
$$
{\rm Pr}\left\{\mathtt{BUY}\right\} = \frac{1}{2} = 1 - {\rm Pr}\left\{\mathtt{SELL}\right\}.
$$

If $\Delta_1$ takes $\kappa=3$, a low price will generate a $\mathtt{BUY}$ decision about 95.3\% of the time,  
and a high price will generate a $\mathtt{SELL}$ decision about 95.3\% of the time.  Because the only 
``messages'' are the price differences at each time step, which are always known with certainty to be $\pm 1$, 
we have either $\varphi_{D_t = 1} = 1, \varphi_{D_t = -1} = 0$ or 
$\varphi_{D_t = 1} = 0, \varphi_{D_t = -1} = 1$.  In the former case, the pragmatic 
information is then 
\begin{align*}
\Phi_{\Delta_1} (\left\{+1, -1\right\}; \left\{{\mathtt{BUY}, \mathtt{SELL}}\right\}) &= 
   1 \times \left[ p_{\mathtt{BUY}\ |\ D_t =+1} \log\left(\frac{p_{\mathtt{BUY}\ |\ D_t =+1}}{1/2}\right) + 
    p_{\mathtt{SELL}\ |\ D_t =} \log\left(\frac{p_{\mathtt{SELL}\ |\ D_t =+1|+1}}{1/2}\right) \right] + \cr
& \hspace{0.5 cm} 0 \times \left[p_{\mathtt{BUY}\ |\ D_t =-1} \log\left(\frac{p_{\mathtt{BUY}\ |\ D_t =-1}}{1/2}\right) + 
   p_{\mathtt{SELL}\ |\ D_t =-1} \log\left(\frac{p_{\mathtt{SELL}\ |\ D_t =-1}}{1/2}\right) \right];\cr
\end{align*}
in the latter case, 

\begin{align*}
\Phi_{\Delta_1} (\left\{+1, -1\right\}; \left\{{\mathtt{BUY}, \mathtt{SELL}}\right\}) &= 
   0 \times \left[ p_{\mathtt{BUY}\ |\ D_t =+1} \log\left(\frac{p_{\mathtt{BUY}\ |\ D_t =+1}}{1/2}\right) + 
    p_{\mathtt{SELL}\ |\ D_t =+1} \log\left(\frac{p_{\mathtt{SELL}\ |\ D_t =+1}}{1/2}\right) \right] + \cr
& \hspace{0.5 cm} 1 \times \left[p_{\mathtt{BUY}\ |\ D_t =-1} \log\left(\frac{p_{\mathtt{BUY}\ |\ D_t =-1}}{1/2}\right) + 
   p_{\mathtt{SELL}\ |\ D_t =-1} \log\left(\frac{p_{\mathtt{SELL}\ |\ D_t =-1}}{1/2}\right) \right],\cr
\end{align*}
both of which evaluate to about 0.72 bits per time step.\\

When $\Delta_1$ makes the right decision, which it does much of the time, all of the 0.72 bits of pragmatic 
information gleaned from knowing $D_t$ at each time step constitute pragmatically useful information.  
However, after the sequence of prices 6, 5, 6, 5, 6, 5, 6 at the beginning of the period, $\Delta_1$ will 
make a $\mathtt{SELL}$ decision when 
a $\mathtt{BUY}$ decision would have been more profitable.  The knowledge of the price change leading to that 
decision is therefore pragmatic noise, a conclusion that would escape an analysis based on Shannon 
information.\\

After a while, other market participants notice $\Delta_1$'s profitability, and they decide to do what $\Delta_1$ is doing, 
albeit after $\Delta_1$'s trades are executed.  This time lag prevents these other participants from executing at the same 
prices as $\Delta_1$ because 
$\Delta_1$'s trades have depleted the finite number of shares bid/offered at these prices.  Thus, market participants
such as $\Delta_2$ that
are only able to execute after $\Delta_1$ might see the periodic dollar prices 
$$
6.1, 5.1, 5.9, 5.1, 5.9, 5.1, 5.9, 5.1, 5.9, 6.9, \quad 6.1, 5.1, 5.9, 5.1, 5.9, 5.1, 5.9, 5.1, 5.9, 6.9, \ldots .
$$  
Assuming that $\Delta_2$ adopts the same decision rule and prior probabilities as $\Delta_1$, 
the corresponding  (pragmatically useful) pragmatic information is only 0.59 bits per time step, 
except for the end of the cycle and the second time step of the next cycle.  For these time steps, $\Delta_2$ will again 
make a $\mathtt{SELL}$ decision when 
a $\mathtt{BUY}$ decision would have been more profitable.  The knowledge of the price change leading to that 
decision is therefore, once again, pragmatic noise.\\

We might imagine a series of market participants, similarly endowed with the same decision making process as
$\Delta_1$, each moving the market with their trades.  The price 
differences per time step will tend toward zero when there is positive pragmatically useful information, but they will remain 
when there is pragmatic noise .  We might then characterize such a market as being ``$D_t$ efficient''.  However, a finite state 
machine that can detect the entire pattern and respond appropriately can clearly remain profitable even 
in a $D_t$ efficient market.  Profits from that more powerful finite state machine will, nevertheless, also degrade as more market 
participants utilize that more powerful automaton.\\

Much actual trading does seem to be informed by the simple trading rules of ``technical analysis'' such as trend following 
and mean reversion \cite{Murphy}, which are easily implemented by finite state machines.  Many billions of dollars
in profits continue to be earned via these finite-state-machine-based strategies, along 
with other finite-state-machine-based strategies, such as index arbitrage and the triangular currency 
arbitrage\footnote{This arbitrage arises when, for example, exchanging USD for Japanese Yen 
and then exchanging the Yen for Euro generates a different quantity of Euro than tha direct exchange of 
U.S. Dollars for Euro.}.  It follows that actual financial markets are,
 at least sometimes, not finite state efficient, let alone Turing efficient.\\
   
Yet the opportunities for exploiting such inefficiencies are becoming ever more fleeting as 
the computer systems that exploit them get faster and more 
sophisticated, suggesting that markets are indeed followinig the path towards finite state efficiency described
above.  The triangular currency arbitrage, in particular, provided the present author a first-hand illustration of the move towards 
finite state efficiency.  When he began working as a quant at HSBC in 1992, he learned that a trader had been assigned the full time job of 
identifying and exploiting these triangular arbitrages.  However, some months later, the trader had to be reassigned because 
the arbitrages were no longer regularly exploitable by HSBC's relatively high latency setup.  
Nevertheless, ultra-low latency computer systems are sometimes able to exploit this arbitrage to 
this day.\\

More generally, the exploitation of predictive patterns until they no longer contain actionable excess-return (also known as ``alpha'') 
generating information is a fundamental feature of quantitative finance.  Among the many studies that illustrate such alpha decay is 
\cite{McLeanAndPontiff}, who demonstrate that returns to portfolios identified as alpha generating in an academic study, ``lose their
edge'' by a substantial margin after publication of the study.  Other so-called ``calendar anomalies'', such as a predictable increase in 
equities prices every January (the ``January effect'') \cite{NoMoreJanuaryEffect} and a predictable decrease in equities prices every 
Monday (the ``weekend effect'') \cite{NoMoreWeekendEffect}, have also been arbitraged out of existence.

\subsection{But True Market Efficiency is More Than Finite State Efficiency!}

There is considerable evidence that more powerful models than 
finite state machines are sometimes needed to describe market dynamics.  For example, let $P_t$ be the price of some security at time $t$ and 
$r_{t, \delta t} = \ln\left(\frac{P_{t+\delta t}}{P_t}\right)$.  Successive values of $r_{t, \delta t}$ are well known to be uncorrelated for all values of $\delta t$; however, successive values of $r_{t, \delta t}$
cannot be independent, because $|r_{t, \delta t}|$ is autocorrelated.   
If $|r_{t, \delta t}|$ were a non-deterministic regular language, the behavior of its recognizing non-deterministic 
finite state machine 
would be completely characterized by a finite dimensional transition matrix 
of probabilities of state transitions.  For such models, the autocorrelations of states separated by time $t$ decay 
exponentially with $t$, with the (unique) time scale of the decay determined by the largest eigenvalue of the transition matrix
\cite{Tegmark}.  
In contrast, per \cite{Multifractal}, a central ``stylized fact'' about asset returns, which must be linked to states in any putative
finite state model, is that $|r_{t, \delta t}|$ has temporal correlations that decay as a power law in $\delta t$, implying the absence 
of a single decay rate.  
\\

These observations should not be surprising for several reasons.  First, a defining characteristic of financial markets is that market participants 
are active across all time scales, from sub-millisecond statistical arbitrageurs to Warren Buffet's famous ``buy and hold'' strategy, 
which he had sometimes followed for decades.  And market participants are certainly cognizant of even longer time horizons, 
such as the Dutch Tulip Mania of the early 1600's.  
Furthermore, as we have seen, a more powerful gambler can succeed on a 
sequence in which a finite state gambler does not, so it is possible for a market to be finite state efficient, but not efficient with 
respect to more powerful computational models.  The absence of a single decay rate for the above correlations implies the 
need for such models.  Also, the use of pushdown automata or linear bounded automata in market models allows for much more generality in the market 
patterns that can be recognized.  For example, just as a stack is necessary to recognize the recursive 
nature of a string of 
arbitrarily nested, balanced parentheses, so, too, is a stack useful in recognizing trends within trends.  More generally, a stack can help to differentiate the same pattern in different regimes, given their 
different meanings in bull and bear markets.\\

An investor hoping to incorporate a more elaborate computational framework than that provided by 
a finite state machine faces the following sources of computational intractability (among others):
\begin{enumerate}
\item A profitable close-out can only be effected by a 
profitable limit or market order, which means that the contents of the limit order book 
must also be considered.  The state of a market must therefore include, not just the price of the last transaction, but also the
entire LOB.  And that's not all, because, as noted above, iceberg orders comprise a sigificant fraction of executed orders, 
implying that the ``private'' limit order books comprising all of the iceberg orders for a given security should be considered as
well.  Then there is the further complication that orders are often entered as part of a larger strategy, such as index arbitrage,
which involves trading all of the securities that comprise the index.  It follows that the 
complete characterization of a market state requires the concatenation of the private limit order books of 
\underline{all} securities traded in the market, and the relevant state transition is the change in state 
triggered by a newly received order, modification of an existing order, or cancellation of an order, 
whether it results in a trade 
and corresponding price update, or not.  We conclude that the quantity $|m^+|$ of the 
previous section is likely to be quite large.  This is definitely a problem for any of the automata whose running time scales super-
linearly with $|m^+|$.

\item Markets respond to a wide variety of news items and other events, making them non-stationary, so any  
preparatory processing, such 
as the solution to the grammar induction problem, must be completed before there are substantial changes to the market.
Furthermore, the time for the market to fully respond to such changes varies over many orders of magnitude, from sub-
millisecond order imbalances in limit order books, to multi-year changes in economic fundamentals, suggesting that markets
may not have the time to fully respond to a given shock before the next one appears.
\end{enumerate}

We formalize one argument for the computational intractability of market dynamics as follows:
\begin{theorem}
Given the standard double auction market setup that supports orders that are conditional on whether a previous order was executed at 
a given price\footnote{While the order matching engines of most exchanges do not support such conditional order processing, the smart order 
routing engines of those brokerage firms that support basket trading cannot operate without conditional order processing.}, 
questions such as whether a given order is ever executed, whether orders are executed in a given order, etc. are all 
$\mathtt{PSPACE}$-complete in the quantity $m^+$ defined above.
\end{theorem}

\begin{proof}
The lemma below demonstrates the Turing completeness of a routing engine and order queue subject 
to the same assumptions as the theorem, along with the additional assumption that orders for 
arbitrarily large 
numbers of units are allowed.  We note that Turing completeness still obtains if we intersperse arbitrary orders for other securities 
anywhere in the queue.  Rice's Theorem \cite{Sipser} then shows that, not only is it undecidable whether all orders in the original set get 
executued, but also whether these other orders get executed at all, to say nothing of the price levels of such executions.  
In more realistic 
cases in which both the numbers and sizes of limit orders are bounded, the order routing system 
becomes a linear bounded automaton.  Questions such as whether the system will reach a certain state from a given 
set of inputs then become decidable, if only because the system can only cycle through a finite number of symbol configurations on the tape 
and a finite number of internal states.  However, such questions are, by definition, $\mathtt{PSPACE}$-complete.
\end{proof}

We note that this result is consistent with the (unproven) claim in \cite{Markomata} that order 
processing in double auction markets requires at least the equivalent of a bounded linear automaton.
\\

Here is the statement and proof of the 
\begin{lemma}
A routing engine that
\begin{itemize}
\item supports an LOB that is allowed to hold limit orders for an arbitrarily large number of units,
\item an order routing queue that supports the placement of an arbitrarily long, fixed sequence of limit orders,
\item routing that is conditional on whether a previous order was executed at a given price,
\end{itemize}
is Turing complete.
\end{lemma}
\begin{proof}
We demonstrate the Turing completeness of the action of smart order routing (SOR) on the LOB by making the correspondence
in Table 2 between the SOR/LOB action for the order book for each of the $n$
securities and the various elements of a so-called Minsky machine. (See Appendix C for a definition of 
this Turing equivalent automaton.)
\\

\begin{table}[H]
\caption{A Correspondence between a Minsky Machine and SOR}
\begin{tabularx}{\textwidth}{|l|X|c|}
\hline\hline
{\bf Minsky Machine Element} & {\bf Market Element} & {\bf Notes} \\
\hline\hline
Register $R_i$    & Order book for security $i$     & \  \\
\hline
$\mathtt{INC}(R_i, I_k)$ & Place a resting limit $\mathtt{SELL}$ order for one unit of security $i$ at price $P_i$, 
                                        then place order $k$ when order is executed.& \   \\
\hline
																																   
$\mathtt{DECJZ}(R_i, I_k, I_l)$ & Place a resting limit $\mathtt{SELL}$ order for one unit of security $i$ at price $P_i+\theta$, then place an $\mathtt{IOC\ BUY}$ order for one unit of security $i$ also at price $P_i+\theta$, where $\theta$ is the tick size, or smallest increment above $P_i$ at which the exchange will accept an order. 
                                                  If order is executed at price $P_i$, place order $I_k$.  If it is executed at $P_i+\theta$,
                                                  place order $I_l$.& (1, 2) \\
\hline
Controller &``hard coded'' lookup table of orders, executed as specified by execution of previous order
& (3) \\
\hline
\end{tabularx}
\end{table}

Notes:
\\  
\begin{enumerate}
\item The combination of the resting limit  $\mathtt{SELL}$ order for one unit of security $i$ at price $P_i+\theta$ and the
$\mathtt{IOC\ BUY}$ order for one unit of same security at the same price corresponds to a $\mathtt{DECJZ}$ (decrement 
the register by one and jump on zero):  if there are any existing limit orders for the security at price $P_i$, the 
$\mathtt{IOC\ BUY}$ order will be executed at $P_i$ because $P_i$ is strictly smaller than the limit price of $P_i + \theta$
associated with the order.  This results in one less unit of the security available for sale, as required.  The SOR engine then 
moves on to order $I_k$.  If 
there are no resting limit sell orders for security $i$ at price $P_i$, the $\mathtt{IOC\ BUY}$ order is stlll executed because 
of the iceberg $\mathtt{SELL}$ order for the security at price $P_i + \theta$.  However, the next order is processed is $I_l$, 
rather than $I_k$.

\item The $\mathtt{IOC\ BUY}$ order is an IOC order simply to avoid the possibility that the best execution software of the 
exchange determines that there is a resting limit order for security $k$ at price $P_k$ on some other exchange, even though
there is no such order on the exchange of interest.

\item The matching engines currently provided by most exchanges do not offer the conditional execution envisioned by the 
controller above, but the smart order routing (SOR) systems of individual brokers do, as evidenced by the widespread 
popularity of such basket trading strategies as index arbitrage. (In fact, the SOR systems have much more demanding 
requirements, such as the handling of the complexities of 
``leg risk'', that, for example, only some of the trades in the basket can be executed at required price levels.)
\end{enumerate}
\end{proof}

The picture that emerges from all of this is a computational version of \cite{EfficientImpossibility}. That paper argued 
that prices cannot be perfectly aligned with all available information because, if they did, no one would bear the non-zero 
cost of gathering and processing the market information needed to maintain that alignment.  Instead,
 \cite{EfficientImpossibility} proposes that markets 
tend towards ``an equilibrium degree of disequilibrium'', where the marginal return on gathering/processing market information 
equals the marginal cost of such efforts.  Here, we have considered a specific kind of 
cost, namely the computational cost of processing market information.  Market participants vary widely in their ability 
to bear this cost.  Hence, even though markets do not seem to be finite state efficient, some participants will find a given 
market to be computationally efficient, in that they cannot bear the cost of extracting pragmatically useful 
information from price histories.  At the same time, others (such as Simons) will 
be able to exploit its inefficiencies.\\

Even if and when markets approach finite state efficiency, they can, and likely will remain linearly 
bounded inefficient.  We are forced to make this claim because any truly comprehensive attempt to use 
``all available information'' will be, as suggested in the beginning of this section, 
stymied by the sheer volume of requisite computation, even under any conceivable increase in compute.
Nevertheless, because this intractability is enforced by a pragmatic information extraction 
process similar to that described in Section 6.1, market participants with the most 
computational resources will presumably have an edge.  If this computational arms 
race does not result in true Turing equilibrium, it does seem to be pushing financial markets towards the 
equilibrium degree of disequilibrium postulated by \cite{EfficientImpossibility}.  It is 
in this sense that markets are indeed computationally efficient.\\

\section{Directions for Future Research}

As noted in the introduction, we have defined pragmatic information in the simplest possible setting.  While \cite{Weinberger24} argued for the present definition of 
pragmatic information by considering the sample average of a large number of decisions, these decisions were assumed to be independent of each other and of the inputs to previous decisions\footnote{
As \cite{Weinberger24} observed however, the convergence of the average pragmatic information per decision to the present definition can obtain from message ensembles with a wide class of statistical properties.}.  The 
theory in its current state therefore remains silent about the situation where previous decisions and/or the inputs that prompted those decisions to influence the given decision, precisely the 
situation in which $\Delta$ learns from previous results.  Perhaps this shortcoming could be addressed
by studying how the quantity $D_\Delta(\bm \phi || \mathbf q)$ of Equation \eqref{eq: PhiDecomp}
declines as $\Delta$ improves the prior estimates $\mathbf q$.  
Perhaps the theory could also be generalized to assume that
the input messages $m$ includes a series of instructions that may affect $\Delta$'s subsequent decision making capabilities, perhaps by affecting internal states that $\Delta$ maintains.  For example, $m$ might be instructions to modify the source code in the implementation of $\Delta$.  The take-away from the present paper --- that the computational abilities of the receiver play a major role in determining how much 
pragmatic information the receiver can extract from a message --- no doubt applies in these more general settings as well.
\\

Another direction for future research is to explore the parallels between pragmatic information, the Shannon 
theory, and thermodynamic entropy, especially in relation to the EMH.  In the previous section, we applied the 
pragmatic information paradigm to market estimates of equilibrium prices, observing that these estimates vary in quality with 
the computational power of the receiver.  Here, we note that such estimates are also subject to the receiver's cognitive biases, 
such as loss aversion, that have been so thoroughly documented in the behavioral finance literature (See, for example, 
\cite{ProspectTheory}).  The pragmatic information paradigm suggests that such biases are a kind of pragmatic noise, 
preventing 
a receiver from fully processing market information.  Given that there is an extensive theory of dealing with noise in the Shannon
theory ({\it i.e.} coding theory), is there something along the same lines within the pragmatic information paradigm?
\\

Bubbles and crashes might also usefully be viewed through the lens of pragmatic information.  There is certainly much 
pragmatic noise bandied about at those times!  A related idea, per \cite{Markomata}, is that there 
is an inherent tendency towards a ``rise in complexity of the entire market system'', which implies that more processing 
power is needed to separate pragmatically useful information from pragmatic noise in such cases.  When this additional 
processing power is not available, we might expect an increase in pragmatic noise, as in the example in Section 4.5.  Per 
\cite{ComputationalView}, such pragmatic noise can, itself, trigger a bubble.\\

Finally, we recall the observation that pragmatic information can be interpreted as a kind of ``free information'':  much as the 
free energy is the portion of the total energy of a physical system available to do useful work, so pragmatic information is 
the part of the total information available make (potentially) useful decisions.  Yet pragmatic information is certainly a kind of 
information measure, and information is more closely linked to thermodynamic entropy, the derivative of free energy with 
respect to temperature.  It would be worth understanding this parallel more clearly:  just as non-smooth changes in
physical entropy that characterize physical phase transitions, so non-smooth changes in pragmatic information rates might 
help understand abrupt changes in the information processing abilities of a receiver, including the collective behavior of 
market participants, as reflected in dramatic changes in pragmatic noise.
\\  

\section*{Acknowledgements}
The author would like to acknowledge Profs. Harald Atmanspacher and Herbert Scheingraber for organizing the NATO Advanced Study Institute on Information Dynamics in 1991, where the author first learned of the 
need for a theory of pragmatic information and New York University for a partial stipend to attend the Santa Fe Institute
Conference on Collective Intelligence that encouraged him to complete this work.\\

The author would also like to acknowledge the role of the various versions of the large language models Google Gemini and 
Anthropic Claude in 
\begin{itemize}
\item gathering relevant literature references, 
\item sharpening intuitions leading to the 
discussion of gamblers in the text and their connection to pragmatic information,
\item supplying the argument that identified a sequence, $\mathscr S$, for which 
$\mathrm{dim}_\mathfrak P(\mathscr S) > \mathrm{dim}_\mathfrak T(\mathscr S)$,
\item highlighting 
the additional computational complexity introduced into market dynamics by the limit order book, and
\item suggesting that the processing surrounding the limit order book be modeled by a Minsky machine,
\item understanding the implications of Rice's theorem on order book processing, and
\item checking the manuscript for correctness.  While Claude Sonnet confirmed the basic soundness of the arguments as 
written by the author, it found a few typos and minor gaps in the proofs as written.
\end{itemize}
It goes without saying that the author has carefully checked the line-by-line correctness of these contributions.

\appendix
\appendixpage
\parindent 0px
\section{The Chomsky hierarchy of formal languages}
\parindent 0px
Given the finite alphabet $\mathcal A$ and strings of symbols that are elements of $\mathcal A^*$, as described in the body of this paper, 
the formal language $\mathcal L$ is a subset of $\mathcal A^*$, with membership in $\mathcal L$ determined by a set of 
rules called a formal grammar for $\mathcal L$.  The {\it membership problem} is then the problem of determining whether a given element of
$\mathcal A^*$ is, in fact a member of $\mathcal L$.  Strings in $\mathcal L$ are sometimes called ``words'', even though the 
theory of formal languages assigns no meaning to such strings {\it per se}.\\

A grammar for $\mathcal L$ specifies which elements of $\mathcal A^*$ are members of $\mathcal L$.  Such grammars 
are often characterized via {\it production rules},
which are sets of transformations that map strings to other strings.  
Besides elements of  $\mathcal A^*$, production rules can contain both 
elements of a finite set, $\mathcal N$, of so-called non-terminal symbols, or even strings of 
multiple symbols in $\mathcal N$.  
The first string produced by the rules is the single ``start'' symbol $\Sigma \in \mathcal N$.  The production 
rules then specify how $\Sigma$ and the symbols appearing in the replacements of $\Sigma$ are to be replaced by other strings.  The rules are successively applied to the strings resulting from the results of previous rule applications until only 
elements of  $\mathcal A^*$ remain (which is why the symbols in $\mathcal A$ are sometimes called
terminal symbols, in contrast to the non-terminal symbols in $\mathcal N$.   Note that, in order for these substitutions to be well-defined, $\mathcal N$ must be disjoint from $\mathcal A$.).  
A grammar may 
specify that the rules may be applied in one of several different ways; such grammars are said to be
non-deterministic.  The {\it membership problem} for a given $w \in \mathcal A^*$  and 
a given language  $\mathcal L$ is then the task of determining whether $w$ can be reproduced by the production rules 
that characterize $\mathcal L$'s grammar in at least one way, starting from the start symbol, $\Sigma$. The {\it reachability problem} asks the 
same question, but starting from a given $w^0 \in \mathcal L$, with $w^0 \neq \Sigma$.\\


Given that $\alpha \in \mathcal A$ and $\nu, \nu' \in \mathcal N$, with $\nu'$ not necessarily distinct from $\nu$, 
the Chomsky hierarchy identifies the languages produced by four cases of such rules, namely
\begin{itemize}
\item {\bf regular languages}, for which the rules must have one of the following forms:
\begin{align*}
\nu &\rightarrow \Lambda \\
\nu &\rightarrow \ \alpha \\
\nu &\rightarrow \ \alpha \ \nu' \\
\nu &\rightarrow \ \nu' \\
\end{align*}
{\it Example}:  
\begin{itemize}[label=\textopenbullet	]
\item $\alpha \in \mathcal A =\{a, b\}$ \\
\item $\nu \in \mathcal N =\{\Sigma, N$ \}\\
\item $\mathcal L = \{w \in \mathcal A^*| w$ contains an even number of $a$'s\}.
\end{itemize}

All such words $w$ can be reproduced by the following production rules \cite{ChatGPT}:
\begin{align}
\Sigma &\rightarrow \Lambda \\
\Sigma &\rightarrow b\Sigma \\
\Sigma &\rightarrow aN		\\
N         &\rightarrow bN		\\
N         &\rightarrow a \Sigma 
\end{align}
The only rule that does not add non-terminal characters to the symbol string being constructed is rule 
(A.1).  Successive applications of this rule will terminate the substitution process only if there are no $N$ 
characters in the string.  
If so, the string must already have an even number of $a$'s, and rule (A.1) will not change this fact.  
If substitution continues for such $N$-free strings via rule (A.2), adding a $b$ via 
this rule (A.2) keeps the number of $a$'s even.  The presence of $N$'s represents a state where the 
string has an odd number of $a$'s, a state that is entered by invoking rule (A.3).  However, the substitution process cannot terminate here because of the non-terminal $N$ that now appears in the string.  The string continues to have an odd number of $a$'s after invoking rule (A.4), but, again, the substitution process cannot terminate here.  In fact, the only way the substitution process can terminate after one or more $N$'s appear in the string is via an application of rule (A.5) for each such appearance,
Not only does each application of (A.5) eliminate an $N$ from the string, but it also changes the count of $a$'s in the string from odd to even.\\																																																							
\\
\item {\bf context free languages}, for which the rules must have the following form:
\begin{align*}
\nu &\rightarrow \ \Big(\mathcal A \cup \mathcal N\Big)^*, 
\end{align*}
{\it Example}:  A language that is context free, but not regular.    
\begin{itemize}[label=\textopenbullet	]'
\item $\alpha \in \mathcal A =\{\  ( , \ )\}$ \\
\item $\nu \in \mathcal N =  \{\Sigma\}$ ({\it i.e.} $\Sigma$ is the only non-terminal symbol)\\
\item $\mathcal L = \{w \in \mathcal A^*| w$ has balanced parentheses $\}$
\end{itemize}
All such words $w$ can be reproduced by the following production rules \cite{AHU}:
\begin{align}
\Sigma &\rightarrow \Sigma \Sigma \\
\Sigma &\rightarrow ( \Sigma ) \\
\Sigma &\rightarrow () 
\end{align}
Note that the final result necessarily has balanced parentheses because only balanced parentheses appear in the string as it is being formed.
\\
\item {\bf context sensitive languages}, for which the rules must have one of the following forms:
\begin{align*}
\Sigma &\rightarrow \Lambda \\
A\ \nu\ B &\rightarrow \ A\ C\ B
\end{align*}
where $A, B \in \Big(\mathcal A \cup \mathcal N \setminus \{ \Sigma \} \Big)^*$ and $C \in \Big(\mathcal A \cup \mathcal N \setminus \{ \Sigma \} \Big)^+$, {\it i.e.} at least one $\mathcal A$
character or at least one $\mathcal N$ character other than $\Sigma$ must be present in $C$.  
Such grammars are context sensitive because
$A$ and $B$ provide a context to $\nu$, as the replacement of $\nu$ by $C$ is permissible only within that context.  In contrast, no such context is required for 
the non-terminal on the left side of the production rules for a context free grammar.  Because the string $C$ is not allowed to be empty, the output of every production rule will be a string that is at least as
long as the input. \\
\\
\\
{\it Example:}  A language that is context sensitive, but not context free.
\begin{itemize}[label=\textopenbullet	]
\item $\alpha \in \mathcal A =\{a, b, c\}$ \\
\item $\nu \in \mathcal N = \{B, C, W, Z\}$ \\
\item $\mathcal L = \{w \in \mathcal A^*| w = a^n b^n c^n, n \geq 1 \}$.  In other words, $w$ must begin with at least one $a$, followed by the same number of $b$'s, followed by the same number of $c$'s.
\end{itemize}
All such words $w$ can be reproduced by the following production rules \cite{AHU}:
\begin{align}
\Sigma &\rightarrow a B C \\
\Sigma &\rightarrow a \Sigma B C \\
CB &\rightarrow CZ \\
CZ &\rightarrow WZ \\
WZ &\rightarrow WC \\
WC &\rightarrow BC\\
aB &\rightarrow ab \\
bB &\rightarrow bb \\
bC &\rightarrow bc \\
cC &\rightarrow cc 
\end{align}
Rules (A.9) and (A.10) allow for the blowing up of the start string to $a^nBC(BC)^{n-1}$,  Rules (A.11) thru (A.14) replace  non-terminal strings $CB$ with non-terminal strings $BC$.  The remaining rules replace 
a non-terminal $B$ or $C$ with the corresponding terminal $b$ or $c$, respectively.  
\item {\bf recursively enumerable languages}, which allow for arbitrary transformations
$$
\eta \rightarrow \xi, 
$$
where $\eta, \xi \in \Big(\mathcal A \cup \mathcal N\Big)^*$.  The term ``recursively enumerable'' 
indicates that all valid strings of the language can be enumerated by a Turing machine.  An example of a language that is 
recursively enumerable, but not context sensitive is one of the languages that solve the {\it Post correspondence problem}.  One of several equivalent
statements of this problem is to think about a collection of tiles, each with a string $m_j^T \in \mathcal A^*$ on the top of the 
tile and another string $m_k^B$ on the bottom.  An example, with $m_j^T$ consisting of the single symbol $\alpha_3$ 
and $m_k^B$ consisting of the concatenation of the two symbols $\alpha_7$ and $\alpha_5$ in the eight symbol alphabet
$\{\alpha_1, \ldots, \alpha_8\}$, is
$$
\begin{bmatrix}
\alpha_3 \\ 
\alpha_7 \alpha_5
\end{bmatrix}
$$
We are then given a set of such tiles, such as 
$$
\begin{bmatrix}
\alpha_2 \alpha_4 \\
\alpha_5 \alpha_4 \alpha_1
\end{bmatrix},
\begin{bmatrix}
\alpha_1 \alpha_7 \\
\alpha_2 \alpha_3 \alpha_8
\end{bmatrix},
\begin{bmatrix}
\alpha_2 \alpha_7 \\
\alpha_2 \alpha_3 \alpha_8
\end{bmatrix}
$$
The problem is to determine whether there is an arrangement of the tiles such that the concatenatation of the top strings 
in the order given by the arrangement is equal to the concatenation of the bottom strings, also as indicated by the arrangement.  
Note that such an arrangement is not possible in the example given above because each of the top strings is shorter than its 
corresponding bottom string, so the concatenation of the top strings will be always be shorter than the concatenation of the 
bottom strings.  However, no such shortcut can be found in the general case; hence this general case is undecidable.

\end{itemize}

It can be shown that only some context free languages are regular, only some context sensitive languages are context free,
only some context sensitive languages are context free, and only some recursively enumerable languages are 
context sensitive.
All of the above examples illustrate non-deterministic grammars because the same left side appears more 
than once in the list of production rules.  

\section{The abstract automata that recognize the languages in the Chomsky hierarchy}
Algorithms that can determine whether any given string in $\mathcal A^*$ is in a given language $\mathcal L$ are often called abstract automata.  When such an automaton can make such a determination
for given language $\mathcal L$, the automaton is said to recognize $\mathcal L$.  If $\mathcal L$ has a non-deterministic 
grammar, more than one sequence of production rules can be activated; recognition is assumed if any of these sequences produces the input string.  Thus, automata, as well as grammars and languages, can be either deterministic or non-deterministic.\\

A fundamental result of theoretical computer science \cite{Sipser} is the table in the body of this paper
that specifies which abstract automata can recognize which languages.  The simplest of the above automata is the {\bf finite state machine}, $\mathfrak F$, which can be characterized by

\begin{itemize}
\item a string of input symbols, $m \in \mathcal M \subseteq \mathcal A^*$, to be input one at a time,
\item a finite set of internal states, ${\mathcal S} = \{s_1, s_2, \ldots s_{|\mathcal S|}\}$, one of which is maintained as a ``current state'',
\item an initial state, $s_0 \in \mathcal S$,
\item a subset, $\mathcal S_F$, of $\mathcal S$ that constitutes a set of final states, entry into any one of which constitutes
a recognition event,
\item a transition table, $\mathbb T_{\mathfrak F}$, which maps pairs $(\alpha | \Lambda, s) \in \mathcal A \cup \{\Lambda \} \times \mathcal S$ to a next state, $s' \in \mathcal S$.  The notation $\alpha | \Lambda$
indicates that the transition table can either read a symbol in $m$ to trigger the $s \rightarrow s'$ transition or this transition can occur spontaneously, by ``reading'' the empty string, $\Lambda$, the latter transition being
known as a lambda transition.  A standard result in automata theory (see \cite{Sipser}, for example) guarantees that $\mathfrak F$ is no less powerful if we assume that every pair $(\alpha, s)$ maps to a unique next state $s'$; that is,  $\mathfrak F$ is deterministic.  $\mathfrak F$ is assumed to read successive symbols in $m$ or make lambda transitions, applying the mappings in $\mathbb T_{\mathfrak F}$ in both cases, to determine the corresponding state transitions until a state in $\mathcal S_F$ has been reached.\\
\end{itemize}
The {\bf pushdown automaton}, $\mathfrak P$, augments $\mathfrak F$ with a stack, $\mathbb S_\mathfrak P$.  
Thus, $\mathfrak P$ can be characterized by
\begin{itemize}
\item a start state, finite sets of internal states and final states, similar to those of $\mathfrak F$, 
\item a stack, $\mathbb S_\mathfrak P$, which is an ordered tuple, $(\beta_1, \beta_2, \beta_3, \ldots )$, of a finite, but arbitrarily large number of symbols chosen from some alphabet, $\mathcal B$, which may or may not be the same as $\mathcal A$,
\item an initialization for $\mathbb S_\mathfrak P$ at the start of processing,
\item a symbol, $\beta_0 \in \mathcal B$, said to be the ``top of the stack'' (The stack ordered tuple serves as a memory; the $\beta$'s of which it is comprised can only become involved in processing by making their way to the top of the stack via the $\mathtt{POP}$ operation described below.),
\item a transition table, $\mathbb T_{\mathfrak P}$, which maps quadruples $ (\alpha | \Lambda, \beta_0, \mathbb S_\mathfrak P, s)\in \mathcal A \cup \{\Lambda\} \times \mathcal B \times \mathcal B^* \times\mathcal S$ to one or more triples $(\beta'_0, \mathbb S'_\mathfrak P, s') \in \mathcal B \times \mathcal B^* \times\mathcal S$, where either
\begin{itemize}[label=\textopenbullet	]
\item $\beta_0$ retains its current value and nothing happens to $\mathbb S_\mathfrak P$ (so that $\beta_0 =\beta'_0$ and $\mathbb S_\mathfrak P =\mathbb S'_\mathfrak P$), but $s \neq s'$, or 
\item $\beta_0$ is replaced by some other element of $\mathcal B$, but nothing happens to $\mathbb S_\mathfrak P$ (so that $\beta_0 \neq \beta'_0$ but $\mathbb S_\mathfrak P =\mathbb S'_\mathfrak P$), or
\item $\beta_0$ retains its current value but it is also ``pushed'' onto $\mathbb S_\mathfrak P$, which then becomes the ordered tuple $(\beta_0, \beta_1, \beta_2, \beta_3, \ldots )$, 
 (so that $\beta_0 = \beta'_0$ but $\mathbb S_\mathfrak P \neq \mathbb S'_\mathfrak P$), or
\item $\beta_0$ is replaced by $\beta_1$, which has been ``popped'' from $\mathbb S_\mathfrak P$, which then becomes $(\beta_2, \beta_3, \ldots )$
(so that, in general, $\beta_0 \neq \beta'_0$ and $\mathbb S_\mathfrak P \neq \mathbb S'_\mathfrak P$)
\end{itemize} 
Once again, it is assumed that $\mathfrak P$ reads successive symbols in $m$, applying the mappings in $\mathbb T_{\mathfrak P}$ after each new symbol has been read to determine the corresponding state transitions .
The non-deterministic version of the pushdown automaton is strictly more powerful than the deterministic version, in that there are languages that can be recognized by the former, but not by the latter \cite{AHU}.

\item The {\bf bounded linear Turing machine}, $\mathfrak B$, and the {\bf general Turing machine}, $\mathfrak T$, augment the pushdown automaton with a second stack.  We denote $\mathfrak B$'s second stack as $\mathbb S \mathbb S_\mathfrak B$, and $\mathfrak T$'s second stack as $\mathbb S \mathbb S_\mathfrak T$.  The only difference between $\mathfrak B$ and $\mathfrak T$ is that both of $\mathfrak B$'s stacks can be no longer than a finite multiple of $|m|$ that is specified before processing begins.  It can be shown \cite{AHU} that deterministic Turing machines and linear bounded automata are no less powerful than their non-deterministic counterparts, in that they can recognize precisely the same set of languages.\\

There are many versions of $\mathfrak T$ that can be shown to be equivalent.  One such that might be more familiar to some
readers consists of
\begin{itemize}[label=\textopenbullet	]
\item a finite set of internal states, ${\mathcal S} = \{s_1, s_2, \ldots s_{|\mathcal S|}\}$, one of which is maintained as a ``current state'', 
\item an initial state, $s_0 \in \mathcal S$,
\item a semi-infinite tape ${\mathcal T}$, consisting of ``cells'', $C_0, C_1, \ldots$, which can contain either a blank or a single element of ${\mathcal A}$,
\item a final state, $s_H \in \mathcal S$, where the subscript ``H'' indicates that processing halts upon entry into that state, 
\item a ``read head'' that points to a specific cell in ${\mathcal T}$
\item a mechanism for moving the read head one cell to the left or right, thus positioning it to read this adjacent cell,
\item a transition table, $\mathbb T$, which maps pairs $(\alpha, s) \in \mathcal A \times \mathcal S$ to the triple 
$(\alpha', s', \mathtt{L | R})$ , with $\alpha' \in \mathcal A$, $s' \in \mathcal S$, and $\mathtt{L | R})$ indicating that the read head should move to
either the left or right.
\end{itemize}
\end{itemize}

\section{On the Minsky Machine}

A Minsky machine is an abstract automaton supporting 
\begin{itemize}
\item a finite number of registers, $R_1, R_2, \ldots, R_m$, which are memory cells, each of which can hold an arbitrarily large 
non-negative integer, $s$,
\item a controller, $\mathtt C$, for sequencing instruction processing,  
\item a list of instructions, $I_1, I_2, \ldots, I_n$, which consist of
\begin{itemize}[label=\textopenbullet	]
\item an instruction, $\mathtt{INC}(R_i, I_k)$, that increments register $R_i$ by one and signals the controller to process 
instruction $I_k$ unconditionally, 
\item an instruction, $\mathtt{DECJZ}(R_i, I_k, I_\ell)$, that decrements register $R_i$ by one and signals the controller to 
process instruction $I_k$ 
unless the register subsequently holds zero, in which case the controller is signaled to process instruction $I_\ell$. 

\end{itemize}
\end{itemize}

Remarkably enough, Minsky machines have the same computational capabilities as Turing machines \cite{Minsky}.  This fact
is made plausible because, per \cite{Minsky}, a Minsky machine can do arbitrary arithmetical operations on the contents of 
its registers and a Turing machine tape with $|\mathcal A|$ symbols per cell is effectively the base $|\mathcal A|$ representation of 
a single integer.  For example, here is how 
a two register Minsky machine can multiply the existing contents of register $R_1$ by two after clearing the contents of $R_2$:
\\
\begin{table}[H]
\centering
\begin{tabularx}{\textwidth}{c l X}
\hline
\textbf{Instruction \#} & \textbf{Instruction} & \textbf{Explanation and Notes} \\
\hline
1 & \texttt{DECJZ}($R_2$, 1, 2) & Decrement $R_2$ until it is zero, in preparation for placing the final result there \\
2 & \texttt{DECJZ}($R_1$, 3, 5) & Decrement $R_1$ as long as it is not zero, thus ensuring that instructions 3 and 4 will execute exactly $N$ times \\
3 & \texttt{INC}($R_2$, 4) & \\
4 & \texttt{INC}($R_2$, 2) & \\
5 & \texttt{HALT} & $R_2$ now holds $2N$ \\
\hline
\end{tabularx}
\caption{Minsky Machine code for computing $2N$, for an arbitrary positive integer $N$.}
\label{tab:minsky-2N}
\end{table}                                  


\begin{thebibliography}{99}

\bibitem{Weinberger24} Weinberger, Edward D.,  ``Towards a Theory of Pragmatic Information'', available 
on arXiv as https://arxiv.org/abs/2403.12324.
\bibitem{ShannonAndWeaver62}  Shannon, Claude and Weaver, Warren. {\it The Mathematical Theory of Communication}, University of Illinois Press, 1962.
\bibitem{SpeakingRate} Christophe Coupé {\it et al.} , ``Different languages, similar encoding efficiency: Comparable information rates across the human communicative niche", {\it Science Advances}, {\textbf 5} \#9 (2019).  Accessible on the web as science.org/DOI/10.1126/sciadv.aaw2594
\bibitem{CoverAndThomas} Cover, Thomas M. and Thomas, Joy, A., {\it Elements of Information Theory}, $2^{nd}$ Ed., John Wiley \& Sons, Inc., 2006.
\bibitem{Kullback} Kullback, Solomon, 	{\it Information Theory and Statistics}, Dover Publications, 1997.
\bibitem{Sipser} Sipser, Michael, {\it Introduction to the Theory of Computation}, Third Edition, Cengage Learning, 2012.
\bibitem{Turing} https://cstheory.stackexchange.com/questions/2515/can-a-probabilistic-turing-machine-solve-the-halting-problem
\bibitem{DimensionInComplexityClasses} Lutz, Jack H., ``Dimension in Complexity Classes'', {\it SIAM J. COMPUT.}, {\textbf 32} (5), 1236-1259 (2003).
\bibitem{Pushdown} Doty, David and Nichols, Jared, ``Pushdown Dimension'', {\it Theoretical Computer Science}, {\textbf 381} 
105-123 (2007).
\bibitem{wow} Pierre Baldi, Laurent Itti, ``Of bits and wows: A Bayesian theory of surprise with applications to attention," {\it Neural Networks}, {\bf 23}, Issue 5, 649-666, 2010.
ISSN 0893-6080,
https://doi.org/10.1016/j.neunet.2009.12.007.
(https://www.sciencedirect.com/
science/article/pii/S0893608009003256).
\bibitem{AHU} Hopcroft, John E., Motwani, Rajeev, Ullman, Jeffrey D., {\it Introduction to Automata Theory, Languages, and
Computation}, $2^{nd}$ Ed., Addison Wesley, 2001.
\bibitem{ValiantParsing}Valiant, Leslie, General context-free recognition in less than cubic time. Carnegie Mellon University. Journal contribution. https://doi.org/10.1184/R1/6605915.v1, 1974.
\bibitem{PSPACE}Arora, Sanjeev; Barak, Boaz, {\it Computational Complexity: A Modern Approach}, Cambridge University Press, p. 92, 2009. ISBN 978-1-139-47736-9.
\bibitem{Chaitin} Chaitin, G. J., "A Theory of Program Size Formally Identical to Information Theory", {\it J. Assoc. Comput. Mach.} {\bf 22}, 329-340, 1975.
\bibitem{CrutchfieldID} Crutchfield, James.  ``Reconstructing Language Hierarchies'', in {\it Information Dynamics}, ed. by H. A. Atmanspracher, {\it et al}, Plenum Press, New York, 1991.
\bibitem{LutzIndividualDimension}  Lutz, J. H., ``The dimensions of individual strings and sequences''. {\it Information and Computation'}, {\textbf 187}(1):49-79 (2003).
\bibitem{Schnorr} Schnorr, C. P.,  ``A unified approach to the definition of random sequences",  {\it Mathematical Systems Theory}, {\textbf 5} (3), 246-258 (1971).
\bibitem{MLRandomness} Martin-L\"of, Per., ``The Definition of Random Sequences'', {\it Information and Control}, {\textbf 9} 602-619 (1966).
\bibitem{Champernowne} Champernowne, D. G. ``Construction of decimals normal in the scale of ten'', {\it J. London Math. Soc.}, {\textbf 2}(8):254-260, 1933.
\bibitem{Maymin} Maymin, Philip Z., "Markets are efficient if and only if P= NP." {\it Algorithmic Finance} {\bf1.1}, 1-11 (2011).
\bibitem{Fama70} Fama, Eugene F.  ``Efficient Capital Markets: A Review of Theory and Empirical Work'' {\it Journal of Finance}, {\textbf XXV}, No. 2, pp. 383-417 (1970).
\bibitem{Fama91} Fama, Eugene F.  ``Efficient Capital Markets: II'' {\it Journal of Finance}, {\textbf XLVI}, No. 5, pp. 1575-1617 (1991).
\bibitem{Murphy} Murphy, John J., {\it Technical Analysis of the Financial Markets:  A Comprehensive Guide to Trading Methods 
and Applications}, New York Institute of Finance, 1999.
\bibitem{ChatGPT} OpenAI, {\it ChatGPT-4}  [Large language model] https://chatgpt.com/c/6785dad4-7b7c-8010-a097-033debbacdb2;  Prompt was "What is a set of production rules for the regular language consisting of strings consisting only of the characters "a" and "b" such that the string has an even number of "a''s?
\bibitem{Multifractal} Jiang, Zhi-Qiang, {\it et. al.}, ``Multifractal analysis of financial markets'', {\it Reports on Progress in Physics}, {\bf 82} (12), 125901 (2019).
\bibitem{ComputationalView} Hasanhodzic, J., Lo, A.W., Viola, E. ``A computational view of market efficiency'', {\it 
Computational Finance}, {\bf 11} (7) 1043-1050 (2011).
\bibitem{ChaitinConstant} Weisstein, Eric W. ``Chaitin's Constant'', From MathWorld--A Wolfram Resource. https://mathworld.wolfram.com/ChaitinsConstant.html, retrieved on January 11, 2026 at 8:17 PM EST/
\bibitem {Weizaeker} Weizsäcker, E.U. von, ``Erstmaligkeit und Bestätigung als Komponenten der
pragmatischen Information'', in {\it Offene Systeme} {\textbf I}, ed. by E.U. von Weizsäcker, 82-113, Klett, 1974.
\bibitem{Tegmark} Lin, Henry W. and Tegmark, Max, ``Criticality in Formal Languages and Statistical Physics'', 
{\it Entropy}, {\bf 19}, 299 (2017).  Available on the web as http://www.mdpi.com/1099-4300/19/7/299.
\bibitem{SFIonDoubleAuction} Friedman, Daniel and Rust, John, editors. {\it The Double Auction Market: Institutions, 
Theories, and Evidence}, Proceedings Volume XIV, Santa Fe Institute Studies in the Sciences of Complexity, Taylor \& Francis,
1993.  ISBN 0-201-62459-1.
\bibitem{NYSEReports} Trading at NYSE (Rule 605 Statistics).  https://www.nyse.com/trade/reports.
\bibitem{WSJ} https://www.wsj.com/finance/stocks/gamestop-mania-highlights-shift-to-dark-trading-11613125980.
\bibitem{Iceberg} https://arxiv.org/pdf/1909.09495.
\bibitem{Gold} Gold, E. M., ``Complexity of automaton identification from given data'', {\it Information and Control}, {\bf 37}, 302–320 (1978).
\bibitem{POMDB} Papadimitriou, Christos H., and John N. Tsitsiklis, “The Complexity of Markov Decision Processes,” 
{\it Mathematics of Operations Research} {\bf 12}, no. 3 (1987): 441–50. http://www.jstor.org/stable/3689975.
\bibitem{CommunicatingMachines} Brand, Daniel and Zafiropul o, Pitro, ``On Communicating Finite-State Machines'', {\it Journal of the ACM} {\bf 30} No. 2323–342 (1983).
\bibitem{NoFreeLunch} Barberis, Nicholas and Thaler, Richard, ``A Survey of Behavloral Finance'', in {\it Handbook of the Economics of Finance}, Elsevier, Amsterdam, pp. 1051–1121.
\bibitem{groping} Negishi, Takashi,``Tâtonnement and Recontracting'', in {\it The New Palgrave Dictionary of Economics}, $3^{rd}$ Ed., Palgrave, Macmillan, pp 
\bibitem{Minsky} Minsky, Marvin, {\it Computation: Finite and Infinite Machines}, $1^{st}$ Ed., Englewood Cliffs, New Jersey, USA: Prentice-Hall, Inc. p. 214  (1967).
\bibitem{McLeanAndPontiff} Mclean, R. David and Pontiff, Jeffrey, ``Does Academic Research Destroy Stock Return Predictability?'', {\it The Journal of Finance}, 
{\bf 71}, No. 1, 5–32 (2016).
\bibitem{NoMoreWeekendEffect} Robins, Russell P. and Smith, Geoffrey Peter, ``No More Weekend Effect''. {\it Critical Finance Review} {\bf 5} 417-424 (2016).
Available at SSRN: https://ssrn.com/abstract=2610172
\bibitem{NoMoreJanuaryEffect} Robins, Russell P. and Smith, Geoffrey Peter, ``On Unmodeled Breaks in the Turn of the Year, Turn of the Month, and January Effects'', 
{\it Financial Review} {\bf 52} No. 4. 725-747 (2017).													  
\bibitem{EfficientImpossibility} Grossman, Sanford J. and Stiglitz, Joseph E., ``On the Impossibility of Informationally
Efficient Markets'', {\it The American Economic Review}, {\bf 70}, No. 3, pp. 393-408 (1980).
\bibitem{Markomata} Mirowski, P., ``Inherent Vice: Minsky, Markomata, and the tendency of markets to undermine 
themselves'',  {\it Journal of Institutional Economics}, {\bf 6}, pp. 415–443 (2010). 
\bibitem{ProspectTheory} Kahneman, Daniel and Tversky, Amos, ``Prospect Theory: An Analysis of Decision under Risk'',
{\it Econometrica}, {\bf 47}, No. 2, pp. 263-292 (1979). 
\end{thebibliography}
\end{document}